\documentclass[%
 reprint,
superscriptaddress,
 amsmath,amssymb,
 aps,
pra,
]{revtex4-2}

\usepackage{graphicx}
\usepackage{dcolumn}
\usepackage{bm}
\usepackage[dvipsnames]{xcolor}
\usepackage{physics}
\usepackage{orcidlink}
\usepackage{mathrsfs}
\usepackage{hyperref}
\hypersetup{
    colorlinks = true,
    citecolor = magenta,
    linkcolor = magenta,
}

\usepackage{algorithm}
\usepackage{algpseudocode}
\usepackage{amsthm}
\usepackage{braket}
\usepackage[english]{babel}
\usepackage{enumitem}
\newtheorem{theorem}{Theorem}

\newtheorem{proposition}{Proposition}
\newtheorem{lemma}{Lemma}
\newtheorem{corollary}{Corollary}

\theoremstyle{definition}
\newtheorem{definition}{Definition}

\newif\ifenablecommands  
\enablecommandsfalse

\begin{document}

\preprint{APS/123-QED}

\newcommand{\SL}[1]{\textcolor{blue}{\textbf{SL}: #1}}



\title{Classically Sampling Noisy Quantum Circuits in Quasi-Polynomial Time under Approximate Markovianity}



\author{Yifan F. Zhang}
\email{yz4281@princeton.edu}
\thanks{These authors contributed equally}
\affiliation{Department of Electrical and Computer Engineering, Princeton University, Princeton, NJ 08544}

\author{Su-un Lee}%
\email{suun@uchicago.edu}
\thanks{These authors contributed equally}
\affiliation{Pritzker School of Molecular Engineering, The University of Chicago, Chicago, IL 60637, USA}

\author{Liang Jiang}
\email{liang.jiang@uchicago.edu}
\affiliation{Pritzker School of Molecular Engineering, The University of Chicago, Chicago, IL 60637, USA}
 
\author{Sarang Gopalakrishnan}
\email{sgopalakrishnan@princeton.edu}
\affiliation{Department of Electrical and Computer Engineering, Princeton University, Princeton, NJ 08544}
\date{\today}

\begin{abstract}
While quantum computing can accomplish tasks that are classically intractable, the presence of noise may destroy this advantage in the absence of fault tolerance.
 In this work, we present a classical algorithm that runs in $n^{\rm{polylog}(n)}$ time for simulating quantum circuits under local depolarizing noise, thereby ruling out their quantum advantage in these settings.
 Our algorithm leverages a property called \emph{approximate Markovianity} to sequentially sample from the measurement outcome distribution of noisy circuits.
We establish approximate Markovianity in a broad range of circuits:
 (1) we prove that it holds for \emph{any} circuit when the noise rate exceeds a constant threshold, and
 (2) we provide strong analytical and numerical evidence that it holds for random quantum circuits subject to \emph{any constant} noise rate.
 These regimes include previously known classically simulable cases as well as new ones, such as shallow random circuits without anticoncentration, where prior algorithms fail.
 Taken together, our results significantly extend the boundary of classical simulability and suggest that noise generically enforces approximate Markovianity and classical simulability, thereby highlighting the limitation of noisy quantum circuits in demonstrating quantum advantage.


\end{abstract}

\maketitle



\section{Introduction}

Quantum computation is strongly believed to solve certain problems in polynomial time that would require exponential time on classical computers. However, real-world quantum devices inevitably suffer from noise, which can degrade or even eliminate such advantages. While fault-tolerant quantum computing can, in principle, sufficiently suppress noise only with polylogarithmic overhead~\cite{shor1996fault, kitaev2003faulttolerant, aharonov2008faulttolerant}, achieving fault tolerance remains a significant technological challenge.

Currently, quantum devices operate in the noisy intermediate-scale quantum (NISQ) regime, characterized by a limited number of qubits and non-negligible noise levels~\cite{preskill2018quantum}. Understanding the computational capabilities of these devices, and identifying tasks where they may still offer an advantage, is therefore crucial. Several experimental works have attempted to demonstrate quantum advantage on near-term hardware~\cite{arute2019quantum,zhong2020quantum,kim2023evidence,morvan2024phase,zhong2021phase,deng2023gaussian,liu2025robust,decross2025computational,king2025beyond}; however, subsequent studies have challenged these claims by developing classical simulation algorithms that can efficiently perform the same tasks~\cite{gao2018efficient,huang2020classical,pan2021simulating,pan2022solving,pan2022simulation,liu2021closing,tindall2024efficient,beguvsic2023fast,oh2022classical,oh2024classical,mauron2025challenging,tindall2025dynamics}.
%
Often, these classical simulation algorithms exploit the presence of noise in NISQ devices: noise can drive the measurement distribution toward trivial forms~\cite{aharonov1996limitations}, or reduce the effective Hilbert space dimension~\cite{aharonov2023polynomial,schuster2024polynomial}, thereby enabling efficient classical simulation.

%


Therefore, it is important to understand the computational power of noisy quantum circuits in the absence of fault tolerance. In particular, given the potential for exponential speedup, we focus on the asymptotic scaling of the computational complexity for simulating noisy quantum circuits. This leads us to the following question: under what conditions can noisy quantum circuits generate output distributions that are classically hard to sample from, in the sense that classical simulation requires resources scaling exponentially with system size? We refer to this as the presence of an \emph{asymptotic quantum advantage}.

In this work, we provide a criterion under which noisy quantum circuits can be efficiently simulated classically, thereby ruling out asymptotic quantum advantage. Specifically, for geometrically local quantum circuits with local depolarizing noise, we show:
\begin{enumerate}
    \item There exists a quasi-polynomial-time classical sampling algorithm for \emph{any circuit} when the noise rate is above a constant threshold.
    \item For \emph{typical random circuits}, strong analytical and numerical evidence suggests that any constant noise rate already suffices for a quasi-polynomial-time classical sampling algorithm.
\end{enumerate}

Our classical algorithm samples qudits sequentially by computing conditional distributions given previously sampled outcomes. This requires a key property called the \emph{approximate Markovianity}:
Under approx. Markovianity, each qubit’s measurement outcome is conditionally independent of qubits outside its immediate neighborhood of some size $\xi$, which we call the \emph{Markov length}. Therefore, the conditional probability of qudit $i$ is fully fixed by the local marginal on a region of radius $\propto \xi$ around qudit $i$. This marginal is a local observable: without fault-tolerance, noisy circuits generate non-trivial distributions only when they are shallow. Therefore, the conditional probability can be computed efficiently using the existence of a light cone. The full output distribution can then be computed by sequentially computing conditional distributions given previously sampled outcomes.

%
%

We establish approximate Markovianity in two settings. First, we prove that the measurement distribution of \emph{any circuit} satisfies approximate Markovianity when the noise rate is above a threshold, using a generalization of the cluster expansion method—a mathematical physics technique that exploits locality to derive a Taylor expansion of the measurement distribution. This allows us to show the decay of conditional mutual information (CMI), an information-theoretic quantity that reflects approximate Markovianity. On the other hand, there exist worst-case shallow circuits that violate approximate Markovianity when the noise rate is low. These circuits exhibit evidence of complexity-theoretic hardness. Therefore, the breakdown of approximate Markovianity is closely related to computational hardness.


The failure of approximate Markovianity in worst-case circuits with weak noise arises from fine-tuned structures within these circuits. In contrast, random quantum circuits typically lack such structure and are therefore more susceptible to noise. We provide both analytical and numerical evidence that approximate Markovianity holds in random quantum circuits subject to any constant noise. Analytically, we derive a novel fourth-moment bound for approximate Markovianity in noisy random circuits and map it to a statistical mechanics model. This model relates approximate Markovianity to the decay of correlations in a biased ferromagnet, which we can rigorously establish in the limit of large Hilbert space dimension. To supplement the analytic results, we present Clifford numerics that show remarkable agreement, even though they operate far from the large Hilbert space limit.


Taken together, these results show that sufficiently strong noise precludes exponential quantum advantage for \emph{worst-case} circuits, while even constant noise appears to eliminate quantum advantage for \emph{average-case} circuits. These findings highlight fundamental limitations of noisy quantum devices in achieving asymptotic advantage.

We compare our results with previous work below. In particular, our algorithm works in regimes where previous algorithms work, and in many cases, we believe our algorithm works in a wider range of circuits. Therefore, we provide a unified algorithm that rules out asymptotic quantum advantage in noisy quantum circuits.

\subsection{Relation to Previous Work}\label{sec:comparison}

In this section we review relevant prior works on the classical simulation of noisy quantum circuits. In particular, we put our algorithm in the context of previous works in two regimes: (i) arbitrary circuits with noise above a threshold, and (ii) typical random circuits with any constant noise rate. We summarize these comparisons in Fig.~\ref{fig:sota}.

\begin{figure}
    \centering
    \includegraphics[width=0.95\linewidth]{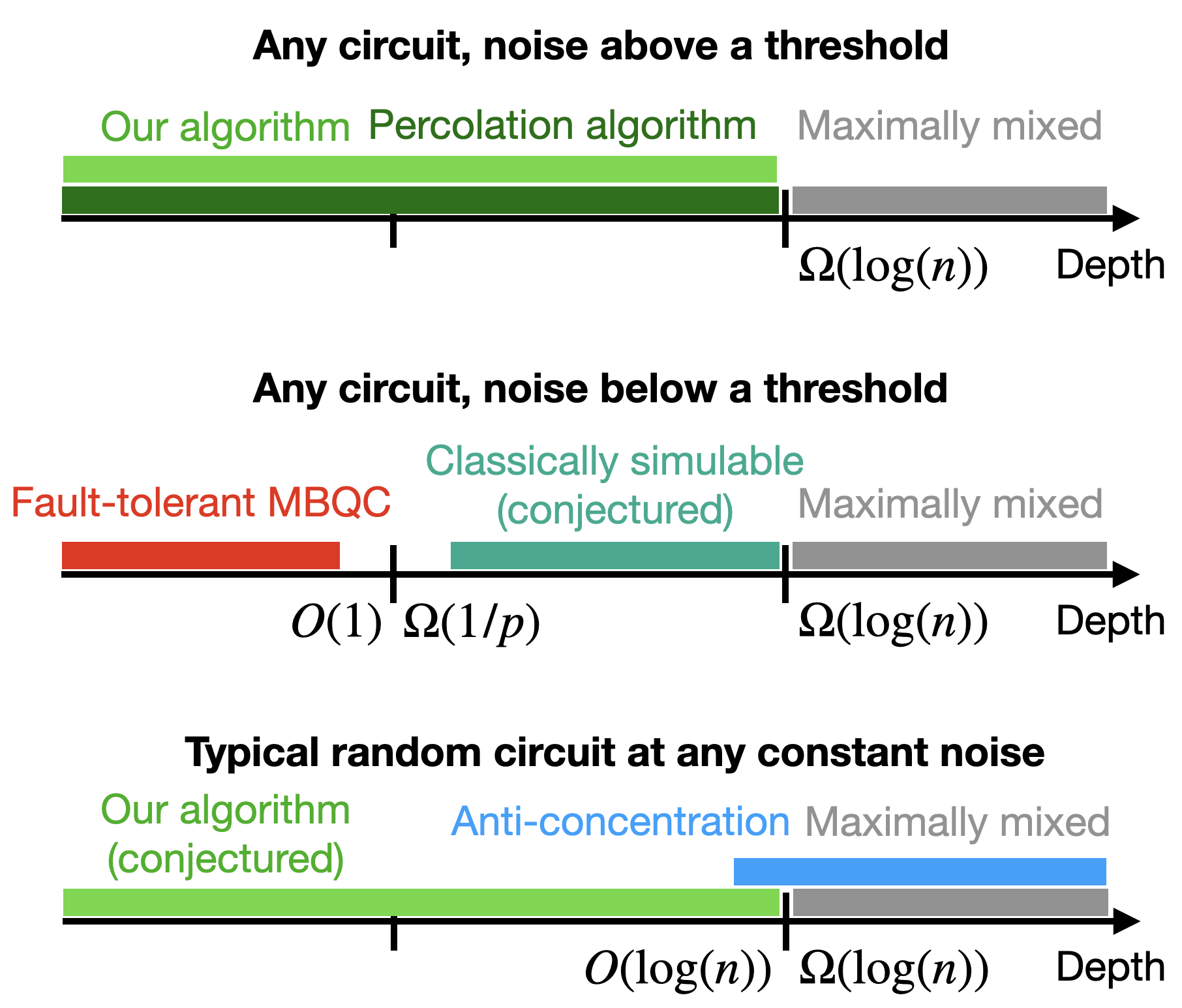}
    \caption{\label{fig:sota} Comparison with previous state-of-the-art classical simulation algorithms for noisy quantum circuits. }
\end{figure}

\subsubsection{Arbitrary Circuit with Noise above a Threshold}

In the regime of arbitrary circuits with noise rate exceeds a constant threshold, several works have proposed efficient classical simulation algorithms based on percolation arguments~\cite{trivedi2021simulability,trivedi2022transitions,rajakumar2025polynomial,nelson2024polynomialtimeclassicalsimulationnoisy}. Specifically, Refs.~\cite{trivedi2021simulability,trivedi2022transitions} observed that a depolarizing channel with rate $p$ can be viewed as tracing out the affected qudit and replacing it with the maximally mixed state with probability $p$. When $p$ exceeds the percolation threshold determined by the circuit geometry, the entire circuit is broken into disconnected parts of smaller circuits with the sizes of $O(\log(n))$ which is allows quasi-polynomial time classical simulation. Since these smaller circuits can be simulated independently, the algorithm runs in quasi-polynomial time.

While the percolation-based arguments lead to results similar to ours---which says an arbitrary circuit can be simulated in quasi-polynomial time once the noise exceeds a threshold---we remark that our approach provides a broader framework. The percolation argument is inherently geometry-based and insensitive to the specific gate set---it strictly fails once $p$ falls below the percolation threshold, no matter what the gates are. In contrast, our approach applies to any constant noise rate $p$ as long as approximate Markovianity holds. Thus, approximate Markovianity serves as a more fine-grained diagnostic of classical simulability than percolation-based methods.

\subsubsection{Arbitrary Circuit with Noise below a Threshold}

On the other hand, it is believed that classical computation cannot efficiently simulate arbitrary quantum circuits when the noise rate is below a constant threshold. This regime should be distinguished from the standard setting of fault-tolerant quantum computation (e.g., Refs.\cite{shor1996fault,kitaev2003faulttolerant,aharonov2008faulttolerant}), as we do not allow for mid-circuit measurements nor the introduction of fresh ancilla qudits during computation. Nevertheless, there has been strong evidence that these quantum circuits can still perform classically hard computations. For example, under plausible complexity-theoretic assumptions, Refs.~\cite{fujii2016computational,fujii2016noisethresholdquantumsupremacy} showed that no efficient classical algorithm can \textit{exactly} sample from the measurement distribution of certain worst-case shallow circuits with a constant noise rate, and Ref.~\cite{bremner2017achieving,bergamaschi2024quantum,rajakumar2024gibbssamplinggivesquantum} showed that approximate sampling is also hard when the noises are modeled as dephasing channels. These results are achieved by embedding classically hard quantum circuits into falt-tolerant measurement-based quantum computing (MBQC), or, error-detecting codes in general, that can tolerate a constant level of noise. In addition, Ref.~\cite{zhang2025conditional} demonstrated that a family of circuits in this regime naturally violates approximate Markovianity, reflecting the persistence of computational hardness in these worst-case circuits.

\subsubsection{Random Circuit with Constant Noise Rate}

Now we turn to the regime of typical random circuits (average-case circuits) with any constant noise rate. While sampling from the noiseless random circuit is proved to be classically hard under plausible complexity-theoretic assumptions~\cite{bouland2018quantum}, the stability of this hardness in the presence of noise has been questioned by several recent works.

Seminal studies~\cite{aharonov2023polynomial,gao2018efficient} demonstrated that noisy random circuits can be efficiently sampled by classical algorithms, provided that the output distribution exhibits \emph{anticoncentration} (see also Ref.~\cite{schuster2024polynomial}). Anticoncentration ensures that the measurement distribution is not concentrated on a small subset of outcomes, thereby allowing these algorithms to approximate the distribution using only low-degree Fourier coefficients~\cite{gao2018efficient}, or a limited number of ``Pauli paths''\cite{aharonov2023polynomial,schuster2024polynomial}. Because random circuits are known to anticoncentrate only at depths of $\Omega(\log n)$\cite{dalzell2022random,deshpande2022tight}, these algorithms are effective in that regime.

However, these anticoncentration-based results do not extend to shallower circuits of depth $o(\log n)$, for which the classical simulability remains unresolved. The shallow-depth regime has drawn particular attention because, while noiseless random circuits are conjectured to be classically intractable beyond a constant critical depth~\cite{napp2022efficient,bene2025quantum}, near-term quantum devices can implement these circuits with relatively high fidelity. Consequently, it remains an open question whether noisy random circuits in this regime can maintain quantum advantage at small but constant noise levels~\cite{cheng2023unraveling}.

In this context, our results provide pessimistic evidence that noisy random circuits are classically simulable in quasi-polynomial time at any constant noise rate, including in the shallow-depth regime. In contrast to previous approaches, our algorithm does not rely on anticoncentration; instead, it exploits the approximate Markovianity, which we expect to hold generically at constant noise levels. Although a rigorous proof of approximate Markovianity in this regime is left for future work, we provide strong analytical and numerical evidence supporting its validity. Together, these findings suggest that noisy random circuits may fail to exhibit asymptotic quantum advantage at any fixed, nonzero noise rate.

\subsubsection{Arbitrary Noisy Circuit with Depth above a Threshold}
Finally, for any constant amount of noise, it is believed that above a critical depth of $O(1/p)$ where $p$ is the noise rate, even worst-case circuits become classically simulable. Intuitively, this is because noise accumulates after each layer of gates. If there are too many layers, then too much noise is injected into the system. This has been proven in restricted classes of instantaneous quantum polynomial (IQP) circuits~\cite{rajakumar2025polynomial} and Clifford circuits with magic initial states~\cite{nelson2024polynomial}. Ref.~\cite{nelson2025limitationsoflocalcirucit}, which will be published concurrently with this paper, makes further technical progress towards worst-case circuits (see Discussion). Conceptually, this is a ``high-noise'' regime except that noises are spread out across layers. Therefore, we conjecture that approximate Markovianity holds in this regime. However, currently we do not have a proof and we leave it to future work.

\subsection{Organization}
This paper is organized as follows. In Sec.~\ref{sec:setup}, we introduce the setup of noisy quantum circuits and the measurement distribution. In Sec.~\ref{sec:sampling}, we present the classical sampling algorithm that samples from the measurement distribution in quasi-polynomial time. In Sec.~\ref{sec:decay_cmi}, we prove the decay of CMI in noisy quantum circuits when the noise rate is above a threshold. In Sec.~\ref{sec:ruc}, we present numerical evidence for approximate Markovianity in random quantum circuits. In Sec.~\ref{sec:comparison}, we compare our algorithm with existing algorithms. Finally, we conclude in Sec.~\ref{sec:discussion}. We present formal proofs and details of analysis in the appendices.

\section{Setup}\label{sec:setup}

\subsection{Noisy Quantum Circuits and Measurement Distribution}

We consider a $D$-dimensional geometrically local noisy quantum circuit composed of $d$ layers of $k$-qubit gates acting on $n$ qudits with local Hilbert space dimension $h$. Starting from the all-zero state $\ket{0}^{\otimes n}$, we apply $d$ layers of unitary $k$-qudit gates. Each layer, described by a unitary matrix $U_j$ for $j=1,\dots,d$, consists of a tensor product of $k$-qudit gates acting on disjoint sets of qudits.

Between the layers of gates, we apply noise channels to all qudits. Specifically, we will mainly consider local depolarizing noise, defined as
\begin{equation}
    \mathcal{N}_{p}[\rho] = (1-p) \rho + \frac{p}{h} I
\end{equation}
where $\frac{1}{h}I$ denotes the maximally mixed state. We apply this channel $\mathcal{N}$ with rate $p$ to all qudits after each layer of gates. Specifically, denoting $\mathcal{N}_{i,j}$ as the depolarizing channel acting on site $i$ after the $j$-th layer, we apply $\otimes_i \mathcal{N}_{i,j}$ after the $j$-th layer of gates.

Finally, we measure all qudits in the computational basis at the end of the circuit. If the output state is $\rho_{\rm out}$, the measurement distribution, denoted by $P$, is given by the diagonal elements of $\rho_{\rm out}$ in the computational basis. In other words, diagonal matrix whose elements are those of $\rho_{\rm out}$ is equivalent to the measurement distribution $P$. This is formally represented by applying a product of completely dephasing channels $\bigotimes_i \mathcal{D}_i$, which removes the off-diagonal matrix elements in the computational basis. Here, for each $i=1,\dots,n$, $\mathcal{D}_i$ denotes completely dephasing channel that acts on the $i$-th qudit. Therefore, without ambiguity, we will use $P$ to denote both the measurement distribution and the corresponding diagonal matrix, i.e., $P = \left( \otimes_i \mathcal{D}_i \right)(\rho_{\rm out})$. For example, a uniform distribution corresponds to $I/h^n$.

Combining all components together, the noisy circuit is represented as a quantum channel $\mathcal{C}$:

\begin{equation}\label{eq:noisy_circuit}
    \mathcal{C} = \left( \otimes_i \mathcal{D}_i \right) \circ \left( \otimes_i \mathcal{N}_{i,d} \right) \circ \mathcal{U}_d  \circ \dots \circ \left( \otimes_i \mathcal{N}_{i,1} \right) \circ \mathcal{U}_1,
\end{equation}
where $\mathcal{U}_j(\cdot) = U_j(\cdot)U_j^\dagger$ denotes a channel that is a product of $k$-qudit gates in one layer (See Fig.~\ref{fig:algo}(a) for an example). Note that we use two legs to represent the bra and ket space of each qudit. For example, an operator $\ketbra{i}{j}$ maps to $\ket{i}\ket{j}$ in the double-leg notation. fixing the double-legs of the second qubit to $\ket{i}\ket{j}$ in  Fig.~\ref{fig:algo}(a) is equivalent to post-selecting the element $\ketbra{i}{j}$ on the second qubit in the density matrices. The grey tensors represent the single-qudit depolarizing channel: it takes one set of double legs as input and outputs one set of double legs. The blue tensors represent the two-qudit unitary channel $\mathcal{U}_j(\cdot)$: it takes two sets of double legs as input and outputs two sets of double legs.

If the unitary gates composing $\mathcal{C}$ are random drawn from the Haar measure independently, we refer to $\mathcal{C}$ as a random circuit. Then the measurement distribution is given by $P = \mathcal{C}[\ketbra{0}{0}]$. With this setup, we denote the marginal probability distribution of $P$ on $X$ as $P_X$, and $P_{X|Y}$ is the conditional distribution of $X$ given $Y$, for the subsets of qudits $X, Y \subset [n]$.

While the local depolarizing noise keeps increasing the entropy of the system, there is no mechanism to decrease the entropy in our setup. First, since we demand that measurements can only happen at the end of the circuits, measurement and feedforward circuits that enable quantum error correction are not allowed. Second, we do not allow introducing clean qudits in the middle of the circuit that can also be exploited to remove the noise in the middle of the circuit.

Therefore, we point out that the circuits we consider have to be \emph{shallow} to produce a non-trivial measurement distribution. Specifically, it was shown in~\cite{aharonov1996limitations} that under local depolarizing noise, the measurement distribution becomes indistinguishable from the uniform distribution when the circuit depth is $d = \omega(\log(n))$.
\begin{proposition}[Rephrased from~\cite{aharonov1996limitations}]
\label{prop:shallow_circuit} Consider a noisy quantum circuit $\mathcal{C}$ of the form Eq.~\eqref{eq:noisy_circuit}. Let $P = \mathcal{C}[\ketbra{0}{0}]$ be the measurement distribution. Then, the total variation distance between $P$ and the uniform distribution is upper-bounded by
\begin{equation}
    \norm{P - \frac{1}{h^n} I}_1 \le n (1-p)^{d}
\end{equation}
In particular, when $d = \omega(\log(n))$, the total variation distance is upper-bounded by $O(1/\rm{poly}(n))$.
\end{proposition}

Because of the above proposition, we will only consider shallow circuits with $d = O(\log(n))$ in this paper. An important consequence of the shallow circuit is that any local observable can be computed in quasi-polynomial time. This is because any local observable depends only on the gates and channels inside its backward light cone, which is controlled by the circuit depth. Under the condition $d = O(\log(n))$ and together with certain constraints on the circuit geometry, we can compute the expectation value of any local observable in quasi-polynomial time.

\subsection{Approximate Markovianity}

We now introduce the approximate Markovianity of the measurement distribution. We formalize this notion by defining the \emph{Markov length} of the measurement distribution below.

\begin{definition}
    Let $P$ be an output probability distribution of a circuit $\mathcal{C}$ on $D$-dimensional lattice. We say that the classical measurement distribution $P$ has a \emph{Markov length} $\xi$ if for any tripartition $A, B, C \subset [n]$,
    \begin{equation}
        \norm{P_{ABC} - P_{AB}P_{C|B}}_1 \le {\rm poly}(n) \cdot \exp(-l_{A,C} / \xi),
    \end{equation}
    where $l_{A,C}$ is the distance between $A$ and $C$. In addition, if $\mathcal{C}$ is a random circuit where the constituting gates are chosen randomly, we say $P$ has a Markov length $\xi$ \emph{on average} if
    \begin{equation}
        \mathbb{E}\left[\norm{P_{ABC} - P_{AB}P_{C|B}}_1 \right] \le {\rm poly}(n) \cdot \exp(-l_{A,C} / \xi),
    \end{equation}
    where $\mathbb{E}[\cdot]$ denotes taking average over the choice of $\mathcal{C}$.
\end{definition}

Using the notion of Markov length, we say a measurement distribution $P$ is \textit{approximately Markovian} if it has a constant Markov length, i.e., $\xi = O(1)$. Note that the definition above considers three partitions that are otherwise unrestricted: $ABC$ together does not have to form the entire system and $B$ does not need to separate $A$ and $C$. This is a strong form of approximate Markovianity and is necessary for our classical sampling algorithm.

We also introduce an information-theoretic quantity called the conditional mutual information (CMI), defined as
\begin{multline}
    I_P(A:C|B)\\
    = H_P(AB) + H_P(BC) - H_P(B) - H_P(ABC).
\end{multline}
Here, $H_P(X)$ is the Shannon entropy of the marginal distribution $P_X$, $H_P(X) = -\sum_{x\in[h]^{|X|}} P_X(x) \log P_X(x)$. The CMI quantifies the amount of correlation between $A$ and $C$ given that the value of $B$ is revealed. If the CMI is small, then $A$ and $C$ are approximately independent given $B$, which is a strong form of approximate Markovianity. This is formalized in the following proposition.
\begin{proposition}[Pinsker's inequality]
    The CMI $I_P(A:C|B)$ of a classical distribution upper-bounds its trace distance to a Markovian distribution
    \begin{equation}
        \norm{P_{ABC} - P_{A|B}P_{BC}}_1 \le 2 \sqrt{I_P(A:C|B)}
    \end{equation}
\end{proposition}

Therefore, if CMI decays exponentially in distance between $A$ and $C$, then the distribution has a finite Markov length and thus satisfies approximate Markovianity. In other words, exponential decay of CMI is a sufficient condition for approximate Markovianity. This is the form of approximate Markovianity that we will establish in Sec.~\ref{sec:decay_cmi}, while we will directly bound $\norm{P_{ABC} - P_{A|B}P_{BC}}_1$ in Sec.~\ref{sec:ruc}.

\section{Sampling Algorithm}\label{sec:sampling}

\begin{figure}
    \centering
    \includegraphics[width=\linewidth]{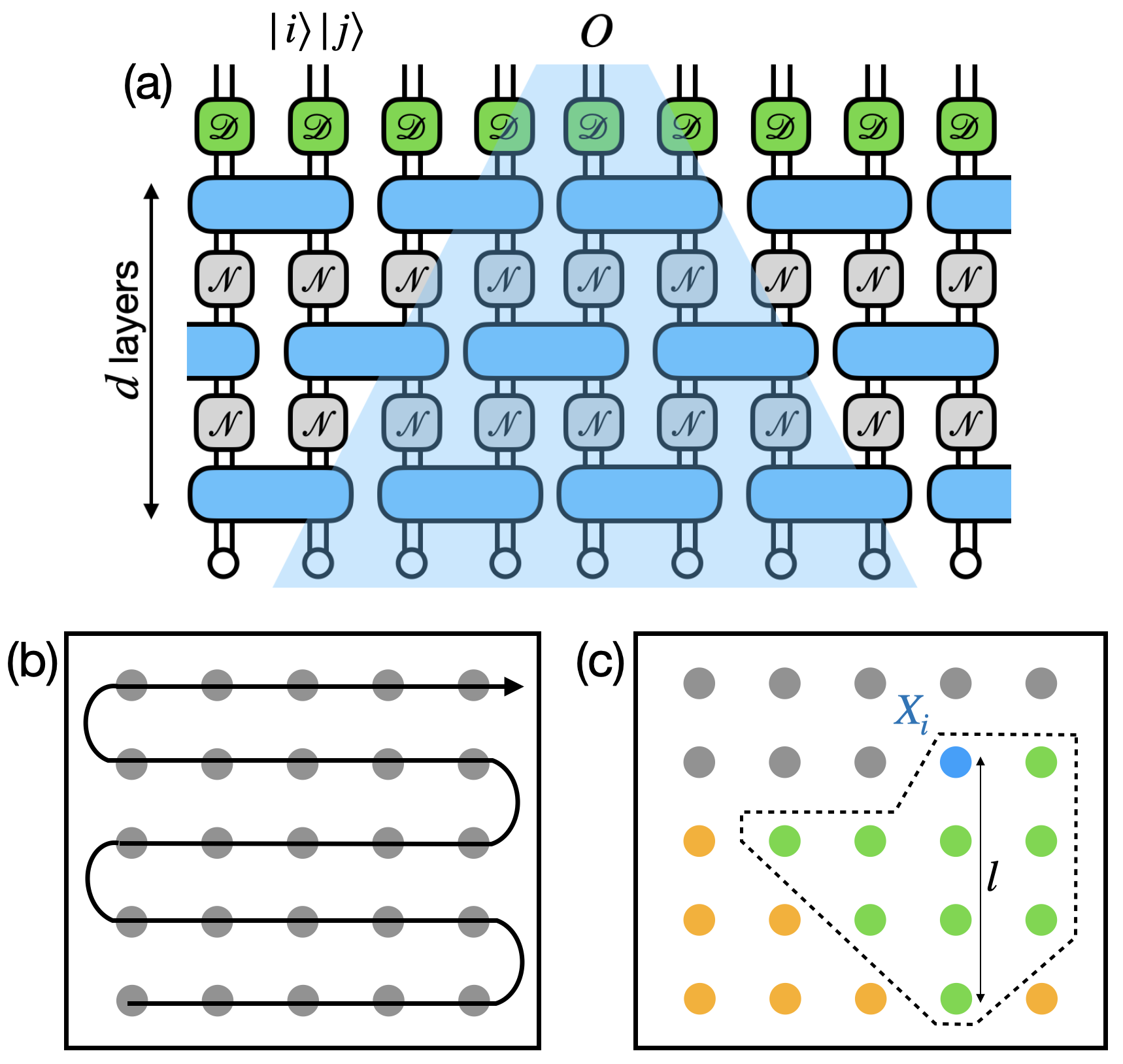}
    \caption{\label{fig:algo} (a) A one-dimensional noisy brickwork circuit with depth $d$. The double legs denote the bra and ket space. The backward light cone of a local observable $O$ is shaded in blue. (b,c) Schematics of our qudit-by-qudit sampling algorithm. (b) One possible sampling path on a two-dimensional grid. (c) The conditional distribution $P_{X_i|B(X_i, l) \cap X_{<i}}$ only depends on the qudits in the ball of radius $l$ around $X_i$ (region circled by the dashed line). We only include qudits in the ball that are already sampled.}
\end{figure}

We now introduce a classical algorithm that samples from the shallow-depth circuit whose output distribution has a constant Markov length. Since shallow circuits allow us to compute local observables, we can efficiently calculate a marginal distribution over a local region. However, this does not directly show that we can sample from the output distribution. To be specific, let random variables $X_1, \dots, X_n \in [h]$ be the measurement outcomes of each qudit, labeled in an arbitrary order [Fig.~\ref{fig:algo}(b)], and $P$ be the joint distribution of them. By Bayes' rule, we have
\begin{equation}
    P = \prod_{i=1}^{n} P_{X_i|X_{<i}},
\end{equation}
where $X_{<i}$ denotes the set of all random variables $X_j$ with $j < i$, i.e., $X_{<i}= \{X_1, \dots, X_{i-1}\}$. This expression enables us to sample from the distribution $P$ by sampling each $X_i$ sequentially, conditioned on the previous variables. However, calculating the conditional probability $P_{X_i|X_{<i}}$ requires an extensive number of qubits, and becomes inefficient as $i$ increases.

This bottleneck of calculating the conditional probability can be greatly alleviated when $P$ has a constant Markov length. Given that $X_1, \dots, X_{i-1}$ are sampled and $P$ has a constant Markov length $\xi$, the conditional distribution $P_{X_i|X_{<i}}$ does not depend on all previously sampled outcomes $X_1, \dots, X_{i-1}$, but is independent of those far away from $X_i$. More specifically, given $l \gg \xi$, those outcomes outside of the $l$-ball around $X_i$ cannot affect the conditional distribution $P_{X_i|X_{<i}}$ [Fig.~\ref{fig:algo}(c)]. With this intuition, we have the following theorem.

\begin{theorem}
    Let $P$ be an output distribution of an $n$-qubit state whose qudits are arranged in a $D$-dimensional grid. Suppose $P$ has a Markov length of $\xi$. Then, for $P'$ defined as
    \begin{equation}
        P' = \prod_{i=1}^{n} P_{X_i|B(X_i, l) \cap X_{<i}}
        \label{eq:markov_decomposition}
    \end{equation}
    with $l = O(\xi \cdot \log(n/\varepsilon))$, $\|P - P'\|_1 \le \varepsilon$. Here, $B(X_i, l) = \{X_j: {\rm dist}(X_i,X_j)\le l \}$ denotes the ball of radius $l$ around $X_i$ In addition, if $P$ arises from a random circuit and has a Markov length of $\xi$ on average, we have $\mathbb{E}[\norm{P-P'}_1] \le \varepsilon$.
    \label{thm:markov_decomposition}
\end{theorem}

\begin{proof}
    For a notational convenience, we denote $B(X_i, l) \cap X_{<i}$ as $B_i$. For $P$ that has a Markov length of $\xi$, we first show the following relation:
    \begin{equation}
        \|P_{X_{<i+1}} - P'_{X_{<i+1}}\|_1 \le \|P_{X_{<i}} - P'_{X_{<i}}\|_1 + c \cdot \exp(-l/\xi).
        \label{eq:thm1-step1}
    \end{equation}
    Where $c = O(\text{poly}(n))$. Noting that $P_{X_{<i+1}} = P_{X_{<i}}P_{X_i|X_{<i}}$ and $P'_{X_{<i+1}} = P'_{X_{<i}}P_{X_i|B_i}$, we have
    \begin{equation}
    \begin{split}
        \|P_{X_{<i+1}} - P'_{X_{<i+1}}\|_1
        &= \|P_{X_{<i}}P_{X_i|X_{<i}} - P'_{X_{<i}}P_{X_i|B_i}\|_1 \\
        &\le \|(P_{X_{<i}} - P'_{X_{<i}})P_{X_i|B_i}\|_1 \\
        &\quad + \|P_{X_{<i}}P_{X_i|X_{<i}} - P_{X_{<i}}P_{X_i|B_i}\|_1 \\
        &\le \|P_{X_{<i}} - P'_{X_{<i}}\|_1 \\
        &\quad + \|P_{X_{<i}}P_{X_i|X_{<i}} - P_{X_{<i}}P_{X_i|B_i}\|_1, \\
    \end{split}
    \end{equation}
    where we used the triangle inequality and the data processing inequality for the first and second inequalities, respectively. Finally, since the distance between $X_i$ and $X_{<i} \setminus B_i$ is larger than $l$, we have Eq.~\eqref{eq:thm1-step1}.

    By applying the telescoping sum using Eq.~\eqref{eq:thm1-step1} with $i=1, \dots, n$, we obtain
    \begin{equation}
        \|P - P'\|_1 \le nc \cdot \exp(-l/\xi).
        \label{eq:thm1-step2}
    \end{equation}
    Finally, choosing $l = O(\xi \cdot \log(n/\varepsilon))$ yields the desired result.

    For $P$ that arises from a random circuit and has a Markov length of $\xi$ on average, we can take the expectation of Eqs.~\eqref{eq:thm1-step1} and \eqref{eq:thm1-step2} over the random circuit and obtain
    \begin{equation}
        \mathbb{E}[\|P - P'\|_1] \le nc \cdot \exp(-l/\xi).
    \end{equation}
    Again, choosing $l = O(\xi \cdot \log(n/\varepsilon))$ yields the desired result.
\end{proof}

Theorem~\ref{thm:markov_decomposition} directly implies that there exists a quasi-polynomial-time classical algorithm that approximately samples from the output distribution of a shallow-depth circuit that has a finite Markov length $\xi$. Specifically, this algorithm aims to get a sample from $P'$ in Eq.~\eqref{eq:markov_decomposition}. Given that $X_1,\dots,X_{i-1}$ have been sampled, calculating $P_{X_i|B(X_i, l) \cap X_{<i}}$ only involves the marginal distribution on $X_i$ and $B(X_i, l) \cap X_{<i}$, which includes $O(l^D)$ qudits. Since calculating this marginal probability can be done by simulating these qudits and their light cone, the runtime of the algorithm is given by $\exp(O((l+d)^D))$.

Let $P$ have a Markov length of $\xi = O(1)$ and the circuit depth $d = O(\log n)$. By choosing $l = O(\log(n/\varepsilon))$, Theorem~\ref{thm:markov_decomposition} gives $\|P-P'\|_1 \le \varepsilon$ and the runtime of the algorithm is given by $\exp(O((\log(n/\varepsilon))^D))$. When $P$ arises from a random circuit and has a Markov length of $\xi = O(1)$ on average, we choose $l = O(\log(n/\varepsilon\delta))$ so that $\mathbb{E}[\|P-P'\|_1] \le \varepsilon\delta$. By Markov's inequality, this algorithm samples from a distribution $P'$ such that $\|P-P'\|_1 \le \varepsilon$ with probability at least $1 - \delta$, in $\exp(O((\log(n/\varepsilon\delta))^D))$ runtime. These results are summarized as follows.

\begin{corollary}
    Let $\mathcal{C}$ be a depth-$d$ quantum circuit of $n$ qudits arranged on a $D$-dimensional grid with $d = O(\log n)$. If the output distribution $P$ has a finite Markov length of $\xi = O(1)$, there exists a classical algorithm that samples from a distribution $P'$ such that $\|P-P'\|_1 \le \varepsilon$, in ${\rm quasipoly}(n,\varepsilon)$ runtime. In addition, if $\mathcal{C}$ is a random circuit and $P$ has a finite Markov length of $\xi = O(1)$ on average, there exists a classical algorithm that samples from a distribution $P'$ such that $\|P-P'\|_1 \le \varepsilon$ with probability at least $1 - \delta$, in ${\rm quasipoly}(n,\varepsilon, \delta)$ runtime.
\end{corollary}

We note that our algorithm utilizing the approximate Markovianity is closely related to previous works. Ref.~\cite{napp2022efficient} used a similar approach to classically simulate 2D shallow-depth random circuits, and Ref.~\cite{brandao2019finite} used quantum analog of Markovianity to generate Gibbs states of local Hamiltonians with an efficient quantum algorithm. Developing upon these, Ref.~\cite{yang2024can} exploited the approximate Markovianity to represent quantum states efficiently with neural networks, and Ref.~\cite{yin2023polynomial} proposed a similar sampling algorithm for classically simulating high-temperature Gibbs states, although their algorithm do not need to exploits approximate Markovianity.

While our algorithm generalizes these previous works with mild modifications, our main technical contribution is to establish approximate Markovianity in several new regimes, which we discuss in the next sections.

\section{Approximate Markovianity in Arbitrary Noisy Circuits}\label{sec:decay_cmi}

In this section, we prove that if the noise rate is above a constant threshold, then any noisy quantum circuit satisfies approximate Markovianity. Specifically, for any tripartition $A$, $B$, and $C$, we show that the CMI $I(A:C|B)$ of the measurement distribution decays exponentially in distance between $A$ and $C$. Although we consider a noisy circuit on $D$-dimensional lattice, we note that the results for CMI decay can be generalized to arbitrary graph-local geometries and we discuss it in Appendix~\ref{app:decay_cmi_high_noise}. Our main result is summarized in the following theorem.

\begin{theorem}\label{thm:decay_cmi_high_noise_informal}
    Consider a $D$-dimensional array of qudits and a noisy quantum circuit $\mathcal{C}$ of the form Eq.~\eqref{eq:noisy_circuit}, where each unitary gate acts on $k$ nearest-neighbor qudits and noise channels are depolarizing channel with noise rate $p$. Then there exists $p_c = 1 - \Omega(1/{\rm poly}(D,k))$ such that for $p\ge p_c$,
    \begin{equation}
        I_P(A:C|B) \le c \, d \min(|\partial A|, |\partial C|) \cdot \exp\left(-l_{AC}/\xi\right),
    \end{equation}
    for any subsystems $A$, $B$, and $C$, where $c=O(1)$, $\xi = k /\ln (q/q_{c})$. Here, $d$ is the number of layers in the circuit, and $|\partial A|$ (resp. $|\partial C|$) are the number of qubits in $A^c$ that are adjacent to $A$ (resp. $C$).
\end{theorem}

The proof is based on the generalized cluster expansion technique, which was originally devised to compute thermodynamic properties in the context of statistical mechanics. There, a Gibbs state $\rho \propto e^{-\beta H}$, with Hamiltonian $H$ at inverse temperature $\beta$, is expanded in powers of $\beta$ with Taylor series. By exploiting the locality of the system Hamiltonian, one can show that this series, or cluster expansion, converges as long as $\beta$ is smaller than some constant threshold (e.g., Refs.~\cite{wild2023classical,haah2022optimal}). In other words, the cluster expansion applies corrections controlled by $\beta$ to the infinite-temperature Gibbs state (which corresponds to $\beta=0$ and is the maximally mixed state).

We adapt this technique to quantum circuits under local depolarizing noise. In our case, the cluster expansion starts from the noise rate $p = 1$ (which results in the maximally mixed state as well) and applies corrections controlled by $1-p$. Analogous to exploiting the locality of the system Hamiltonian in the case of statistical mechanics, we will exploit the locality of the gates in the circuit to show that the cluster expansion converges as long as $p$ is larger than some constant threshold.

\subsection{Proof Sketch}

Now we present the key ideas of our proof, leaving the technical details to Appendix~\ref{app:decay_cmi_high_noise}. To begin with, we denote noise rate for each $\mathcal{N}_{i,j}$ as $p_{i,j}$, and additionally define $q_{i,j} = 1 - p_{i,j}$ for notational convenience (note that we will eventually take $p_{i,j} = p$ and $q_{i,j} = 1-p$). It is also convenient to denote the marginal distribution on a subsystem $L$ as $\bar{P}_L = P_{L} \otimes I_{L^c}$, i.e, $\bar{P}_L$ is the marginal distribution of $P$ on $L$ with identity acting on $L^c$. 

With this setting, the quantity we will focus on in the cluster expansion is the following linear combination of the logarithm of the measurement distribution $P$:
\begin{multline}
    H(A:C|B) \\
    = \log(\bar{P}_{AB}) + \log(\bar{P}_{BC}) -  \log(\bar{P}_{B}) -  \log(\bar{P}_{ABC})
\end{multline}
where $A$, $B$, and $C$ are some tripartition of the system. Here, note that $H(A:C|B)$ is a $2^n \times 2^n$ diagonal matrix. $H(A:C|B)$ is useful to bound the CMI. Specifically, it is straightforward to see that
\begin{equation}
    I_P(A:C|B) = -\Tr[ P H(A:C|B) ].
\end{equation}
Since the right hand side is the convex sum of the eigenvalues of $-H(A:C|B)$, we can further upper-bound it by the operator norm $\norm{H(A:C|B)}$:
\begin{proposition}\label{prop:cmi_matrix_main} CMI is upper bounded by the operator norm of $H(A:C|B)$,
\begin{equation}
    I_P(A:C|B) \le \norm{H(A:C|B)}.
\end{equation}
\end{proposition}

We then apply the cluster expansion to $H(A:C|B)$, which amounts to taking the multi-variate Taylor expansion with respect to $q_{i,j}$, formally written as
\begin{multline}\label{eq:cluster_expansion}
    H(A:C|B) = \left( \sum_{i,j} q_{i,j} \frac{\partial}{\partial q_{i,j}} \right) H(A:C|B)\\
    + \left(\sum_{(i,j) \neq (i',j')} q_{i,j}q_{i',j'} \frac{\partial^2}{\partial q_{i,j} \partial q_{i',j'}} \right) H(A:C|B)\\
    + \left( \sum_{i,j} \frac{q_{i,j}^2}{2} \frac{\partial^2}{\partial q_{i,j}^2} \right)H(A:C|B) + \cdots,
\end{multline}
where all derivatives are evaluated at $q_{i,j}=0$, and afterwards we set $q_{i,j}=q$ for every spacetime location $(i,j)$.

Then we can show that if $A$ and $C$ are far apart in the circuit, all terms except for the high-order terms in the series of Eq.~\eqref{eq:cluster_expansion} are zero. Therefore, $H(A:C|B)$ consists only of high-order terms in $q$. More specifically, we show the following proposition:

\begin{proposition}
    \label{prop:distance_bound} Let $l_{AC}$ be the distance between $A$ and $C$ in the circuit. If $m < l_{AC}$, then
    \begin{equation}
        \frac{\partial^m}{\partial q_{i_1,j_1} \partial q_{i_2,j_2} \dots \partial q_{i_m,j_m}} H(A:C|B) = 0
    \end{equation}
    for any choice of $(i_1,j_1), (i_2,j_2), \dots, (i_m,j_m)$.
\end{proposition}

To further bound the operator norm of $H(A:C|B)$, a significant challenge exists: given $m>l_{AC}$, the number of $m$-th order terms in Eq.~\eqref{eq:cluster_expansion} grows rapidly in $m$, specifically, ${nd \choose m}$. Therefore, having a higher order coefficient $q^m$ is not sufficient to guarantee that the operator norm of $H(A:C|B)$ is small. To overcome this difficulty, we remark that the locality of the circuit reduces the number of non-zero terms. Specifically, we show that the derivative $\frac{\partial^m H(A:C|B)}{\partial q_{i_1,j_1} \dots \partial q_{i_m,j_m}}$ is non-zero only if the set of spacetime locations $(i_1,j_1), (i_2,j_2), \dots, (i_m,j_m)$ are connected by the gates in the circuit. This results in the number of terms in the expansion to be growing at most exponentially in $m$:

\begin{proposition}[Informal]\label{prop:number_of_terms}
     The number of non-zero $m$-th order terms in the expansion of Eq.~\eqref{eq:cluster_expansion} is upper-bounded by $\min(|\partial A|, |\partial C|) \cdot O(b^m)$ for some $b = {\rm poly}(D,k)$.
\end{proposition}

With this proposition, we show that the operator norm of each $m$-th order term decays exponentially in $m$ when $q \le q_c$ for some $q_c = O(1/{\rm poly}(D,k))$. Specifically, denoting $H^{(m)}$ as the sum of all $m$-th order terms in the expansion of Eq.~\eqref{eq:cluster_expansion}, we show the following proposition:

\begin{proposition}[Informal]\label{prop:derivative_bound}
    There exists $q_c = O(1/{\rm poly}(D,k))$ and $c = O(1)$ such that
    \begin{equation}
        \norm{H^{(m)}} \le c \, d \min(|\partial A|, |\partial C|) \cdot \left(\frac{q}{q_c}\right)^m .
    \end{equation}
\end{proposition}

With these propositions, we now prove the main theorem of this section.

\begin{proof}[Informal proof of Theorem~\ref{thm:decay_cmi_high_noise_informal}]
    Applying triangle inequality to Eq.~\eqref{eq:cluster_expansion}, we have
    \begin{equation}
        \norm{H(A:C|B)} \le \sum_{m=0}^{\infty} \norm{H^{(m)}}
    \end{equation}
    By Proposition~\ref{prop:distance_bound}, all the terms $\norm{H^{(m)}}$ for $m < l_{AC}$ vanish. Then by Proposition~\ref{prop:derivative_bound}, we have
    \begin{equation}
    \begin{split}
        \norm{H(A:C|B)} &\le c \, d \min(|\partial A|, |\partial C|) \cdot \sum_{m=l_{AC}}^{\infty} \left(\frac{q}{q_c}\right)^m\\
        &\le c \, d \min(|\partial A|, |\partial C|) \cdot \frac{(q/q_c)^{l_{AC}}}{1 - (q/q_c)},
    \end{split}
    \end{equation}
    given that $q < q_c$. Finally, applying Proposition~\ref{prop:cmi_matrix_main}, we get the desired result.
\end{proof}

We note that in this regime of worst-case circuits when the noise is above a threshold, earlier work has devised a classical simulation algorithm for worst-case circuits when the noise is above a threshold~\cite{trivedi2021simulability,trivedi2022transitions}. There, the authors observe that the depolarization channel with rate $p$ can be understood as tracing out the qudit and replacing it with the maximally mixed state with probability $p$. When the qudit is traced out and replaced, it breaks the connection in the circuit diagram (See Fig.~\ref{fig:disconnection} for an illustration).

To classically simulate the circuit, the authors randomly perform the partial trace and replacement at each spacetime location with probability $p$. When $p$ is above the percolation threshold, the typical circuit configurations consist of disconnected clusters with size $O(\log(n))$, with a failure rate of $O(1/\text{poly}(n))$. Their algorithm exploits this fact and only simulates the typical circuit configurations. In typical circuit configurations, the algorithm simulates individual clusters separately, and since the cluster is at most $O(\log(n))$ in size, the algorithm runs in quasi-polynomial time.

We compare our algorithm with the percolation algorithm. Both algorithms have a quasi-polynomial time in the high-noise regime. However, the percolation algorithm is only sensitive to the circuit geometry, whereas our algorithm considers the effect of both circuit geometry and the gates in the circuits. As an example, the percolation algorithm strictly fails when $p$ is below the percolation threshold. On the contrary, our algorithm works for any $p$ as long as approximate Markovianity is satisfied. As we will see in the next section, this holds true for typical random circuits at any constant noise rate, where the percolation algorithm could fail.

We note that the threshold given by Theorem~\ref{thm:decay_cmi_high_noise_informal} is a very loose bound. It is derived solely from the circuit geometry, without incorporating any information about the gates. Consequently, we expect the actual threshold for a given circuit to be much lower. More generally, we also expect the threshold for approximate Markovianity to differ from the percolation threshold, since the latter depends only on geometry. An open question is to compare these two thresholds—for example, whether the percolation threshold provides an upper bound on the threshold for approximate Markovianity.

\section{Approximate Markovianity in Noisy Random Circuits}\label{sec:ruc}

In this section, we provide analytical and numerical evidence that typical random quantum circuits subject to any noise rate yield approximately Markovian output distributions with high probability.

The analytical evidence is based on bounding the trace distance between the measurement distribution and the closest Markovian distribution. The bound maps to a statistical mechanics problem of decaying correlations in a pinned ferromagnetic Potts model. We are able to solve this model and obtain an analytical bound in the limit where the local Hilbert space dimension $h=\Omega(n)$. In addition, we present Clifford numerics showing that the trace distance decays exponentially in distance. We also observe consistent behaviors in the numerical results and the statistical-mechanical model, even though they operate in different regimes. This gives us strong confidence that the approximate Markovianity holds in generic noisy random quantum circuits.

\subsection{Existing Results on Measurement Induced Entanglement}

We first review the existing result about approximate Markovianity in \emph{noiseless} random circuits. Consider a noiseless random circuit in two or higher dimensions with depth $d$. It was first observed in~\cite{bao2024finite} that when $d$ is above a constant critical depth, then by measuring a bulk region $\textcolor{Green}{B}$, a long-range entanglement between two boundaries $\textcolor{Cerulean}{A}$ and $\textcolor{Goldenrod}{C}$ is induced. This is called \emph{measurement induced entanglement} (MIE) in the literature. Correspondingly, the measurement distribution becomes non-Markovian~\cite{zhang2024nonlocal, lee2024universal}. Before the critical depth, however, MIE decays exponentially in the distance between $\textcolor{Cerulean}{A}$ and $\textcolor{Goldenrod}{C}$, and correspondingly the measurement distribution is approximately Markovian. Ref.~\cite{napp2022efficient} utilizes this fact to develop a classical simulation algorithm based on the boundary matrix product state method. In this algorithm, the bond dimension of the matrix product state is controlled because MIE decays, resulting in a polynomial-time algorithm.

The above results are mostly empirical. Since then, people have attempted to prove the existence of the threshold. Earlier works apply semi-rigorous techniques based on the replica trick and statistical mechanics of random circuits~\cite{bao2024finite,zhang2024nonlocal}, as well as considering simplified circuits such as random Clifford circuits~\cite{bene2025quantum}. Recently, Ref.~\cite{mcginley2025measurement} managed to rigorously prove that above the threshold depth, MIE becomes long-ranged. Their key technical contribution is to lower-bound MIE with an entanglement witness that can be computed on the second moment. This results in an analytical threshold that is larger than the numerical critical point. However, their technique breaks down for mixed states. To the best of our knowledge, it is currently unclear whether MIE can remain long-range in the presence of any constant noise. Our result suggests that MIE is sensitive to noise and becomes short-ranged at any constant noise level.

\subsection{Analytical Bound from Statistical Mechanics}

\begin{figure*}
    \includegraphics[width=0.9\linewidth]{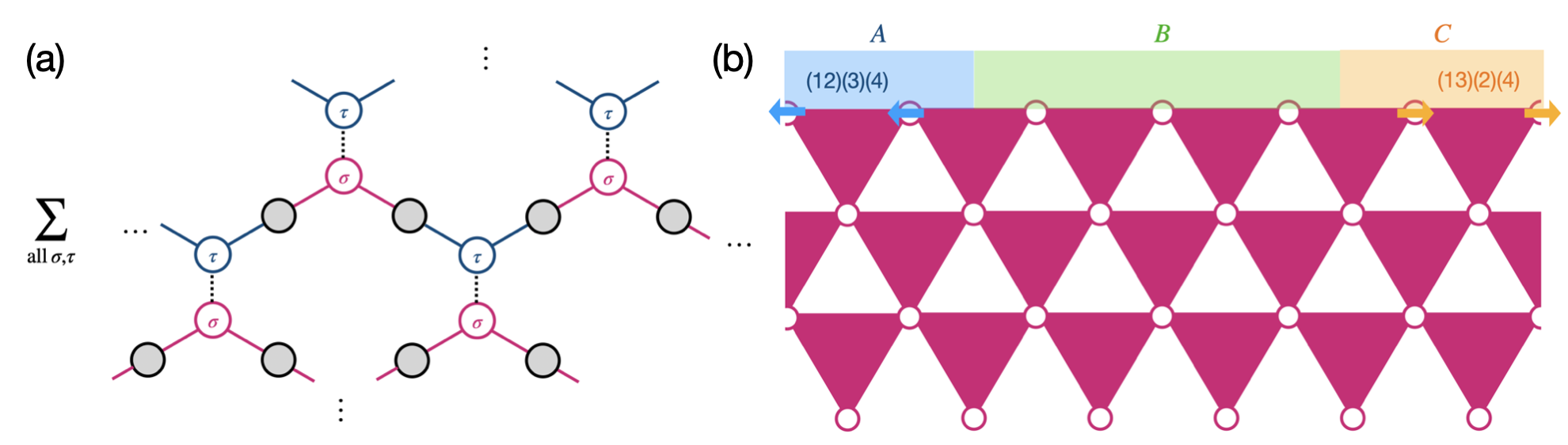}
    \caption{\label{fig:stat_mech} (a) After taking the Haar average, $k$ copies of the noisy random circuit map to an exponential sum of the permutation elements $\sigma$ and $\tau$ at each spacetime location, weighted by the Weingarten function (dashed line). The grey circles represent the depolarizing channel. (b) The statistical mechanics model after integrating out the $\tau$ variables. The boundaries are pinned to elements given in Proposition~\ref{prop:fourth_moment_partition}. }
\end{figure*}

We first discuss the analytical bound. Let each unitary gate in Eq.~\eqref{eq:noisy_circuit} be a random unitary gate sampled from the Haar measure. We denote the final state on $ABC$ \emph{before} the measurement as $\rho_{ABC}$, where we make the dependence of $\rho_{ABC}$ on the choice of the gates implicit. We define the \emph{unnormalized} post-measurement state on $AC$, conditioned on the measurement outcome $b$ on $B$, as
\begin{equation}
    \tilde{\rho}_{AC|b} = \Tr_B[\bra{b} \rho_{ABC} \ket{b}]
\end{equation}
We also let $p_b$ be the probability of measuring $b$ on $B$, and we use $\rho_{AC|b}$ to denote the normalized post-measurement state, i.e. $\rho_{AC|b} = \tilde{\rho}_{AC|b}/p_b$.

Under the above notation, we consider the average trace distance between the post-measurement state $\rho_{AC|b}$ and the product state $\rho_{A|b} \otimes \rho_{C|b}$, where $\rho_{A|b}$ and $\rho_{C|b}$ are the marginal states on $A$ and $C$ respectively. We define this quantity as $\bar{D}$:
\begin{equation}\label{eq:bar_D}
    \bar{D} = \mathbb{E}_{U \sim \mathrm{Haar}} \sum_{b} p_b \norm{\rho_{AC|b} -  \rho_{A|b} \otimes \rho_{C|b}}_1
\end{equation}

One can see that $\bar{D}$ upper-bounds the distance of the measurement distribution $P$ to the closest Markovian distribution by the data-processing inequality on $AC$, after averaging over the Haar measure.

We first state our analytical bound below.

\begin{theorem}
    \label{thm:stat_mech} Consider the state $\rho_{ABC}$ from a noisy random quantum circuit with noise rate $p$ discussed previously. Then the average trace distance $\bar{D}$ in Eq.~\eqref{eq:bar_D} is upper-bounded by
    \begin{equation}\label{eq:bar_D_bound}
        \bar{D} \leq h^{3d (|\partial_A B|+|\partial_C B|)} \Delta Z^{\frac{1}{4}}
    \end{equation}
    Where $|\partial_A B|$ and $|\partial_C B|$ are the sizes of the boundary of $B$ in contact with $A$ and $C$ respectively, and $\Delta Z$ is related to the partition function of a 24-state Potts model, defined in Proposition~\ref{prop:fourth_moment_partition}.
    
Furthermore, Suppose $B$ is $D+1$ dimensional box. When $h > c n$ where c is a constant, $\Delta Z$ is given by
     \begin{equation}\label{eq:bar_D_bound_limit}
         \Delta Z  = O(e^{-l_{AC} / \max(\xi_h,\xi_p)} ),
    \end{equation}
    where $\xi_h = - 1/\log(\frac{n}{ch})$ and $\xi_p = - 1/(\mathcal{A} \log(1-p'))$. $p'=\Theta(p)$. $\mathcal{A}$ denotes the cross-sectional area of $B$ perpendicular to the direction connecting $AC$. In particular, $\mathcal A \propto d$.
\end{theorem}


\subsubsection{Deriving the Bound}

The proof of Theorem~\ref{thm:stat_mech} is based on bounding $\bar{D}$ in terms of the moment properties of the Haar measure, that is $k$ copies of the state. This approach is commonly found in decoupling-type inequalities~\cite{preskill1998lecture}, where the one-norm is upper-bounded by the two-norm squared, which then becomes a second moment property and can be computed using two copies of the state. In our notation, the $k$-th moment property would depend on $k$ copies of the unnormalized post-measurement state $\tilde{\rho}_{AC|b}$. The lack of normalization is crucial because normalization requires division, which makes the state highly non-linear in the Haar measure.

However, there is a subtlety here: while $\rho_{AC|b}$ can be constructed using one copy, $\rho_{A|b} \otimes \rho_{C|b}$ requires two copies of the state. Therefore, constructing $\norm{\rho_{AC|b} -  \rho_{A|b} \otimes \rho_{C|b}}$ already requires two copies of the state. To convert the one-norm to the two-norm, we will therefore need four copies. We state the resulting bound below and defer the proof to the appendix.

\begin{lemma}
    \label{lem:trace_distance_bound} Consider the state $\rho_{ABC}$ from a noisy random quantum circuit discussed previously. Then the average trace distance $\bar{D}$ in Eq.~\eqref{eq:bar_D} is upper-bounded by
    \begin{equation}
        \bar{D} \leq h^{3|AC|} \left(2 h^{4|B|} \mathbb{E}_{U \sim \mathrm{Haar}} p_0^4 \norm{\rho_{AC|0} - \rho_{A|0} \otimes \rho_{C|0}}_2^2 \right)^{\frac{1}{4}}
    \end{equation}
    Where $|AC|$ is the number of qudits in $AC$, $p_0$ is the probability of measuring the all-zero state on $B$, and $\rho_{AC|0}$ is the post-measurement state conditioned on measuring the all-zero state on $B$. Moreover, the quantity in the parenthesis can be computed using four copies of the state $\tilde{\rho}_{AC|0}$.
    \begin{equation}
        \begin{split}
            \mathbb{E}_{U \sim \mathrm{Haar}} p_0^4 \norm{\rho_{AC|0} - \rho_{A|0} \otimes \rho_{C|0}}_2^2  = \mathbb{E}_{U \sim \mathrm{Haar}} \\
            \norm{\tilde{\rho}_{AC|0} \otimes \Tr[\tilde{\rho}_{AC|0}] - \Tr_A[\tilde{\rho}_{AC|0}] \otimes \Tr_C[\tilde{\rho}_{AC|0}]}_2^2
        \end{split}
    \end{equation}
\end{lemma}

The proof of the above lemma is technical but elementary so we defer it to the appendix. It can be understood as a generalization of the decoupling-type inequalities.

It is well-established that the moment properties of random circuits can be interpreted as a statistical mechanics problem. We will present the model here and consider the effect in the presence of noise. Without loss of generality, we will work with brickwork circuits with nearest-neighbor two-qubit gates. We consider two local depolarizing channels $\mathcal{N}_p^{\otimes 2}$ acting on the two qubits, followed by a Haar-random two-qubit gate $U$. 

Let the $h$-dimensional Hilbert space be spanned by $\{\ketbra{i}{j}\}_{i,j=0}^{h-1}$. We transpose the bra space and consider the double Hilbert space $\{\ket{i} \otimes \ket{j}\}_{i,j=0}^{h-1}$. $\textcolor{magenta}{\ket{\sigma}}$ and $\textcolor{RoyalBlue}{\bra{\tau}}$ are defined on this space and are given by
\begin{equation}
    \textcolor{magenta}{\ket{\sigma}}= \sum_{i} \ket{i} \otimes \ket{\sigma(i)} 
\end{equation}
Where $\sigma(i)$ is the image of $i$ under the permutation $\sigma$. $\textcolor{RoyalBlue}{\bra{\tau}}$ is defined similarly.  With this notation, the fourth moment of $U \circ \mathcal{N}_p^{\otimes 2}$ is given by:
\begin{multline}\label{eq:noisy_haar_average_word}
    \mathbb{E}_{U \sim \mathrm{Haar}} (U \otimes U^{*})^{\otimes 4} \circ (\mathcal{N}_p^{\otimes 2})^{\otimes 4} =  \\
     \sum_{\sigma, \tau \in S_4} \mathcal{W}(\sigma^{-1}\tau, h^2) \textcolor{magenta}{\ket{\sigma}^{\otimes 2}} \textcolor{RoyalBlue}{\bra{\tau}^{\otimes 2}} \times \left( p \, \textcolor{gray}{\ketbra{e}{e}} + (1-p) \textcolor{gray}{I} \right)^{\otimes 2}    
\end{multline}
Where 
where $\sigma$ and $\tau$ are elements of the symmetric group $S_4$, and $\mathcal{W}(\sigma^{-1}\tau, h^2)$ is called the \emph{Weingarten function} and is computable analytically~\cite{fisher2022random,collins2021weingarten}. Eq.~\eqref{eq:noisy_haar_average_word} can be represented as a tensor network shown below.

\begin{center}
    \includegraphics[width=0.8\linewidth]{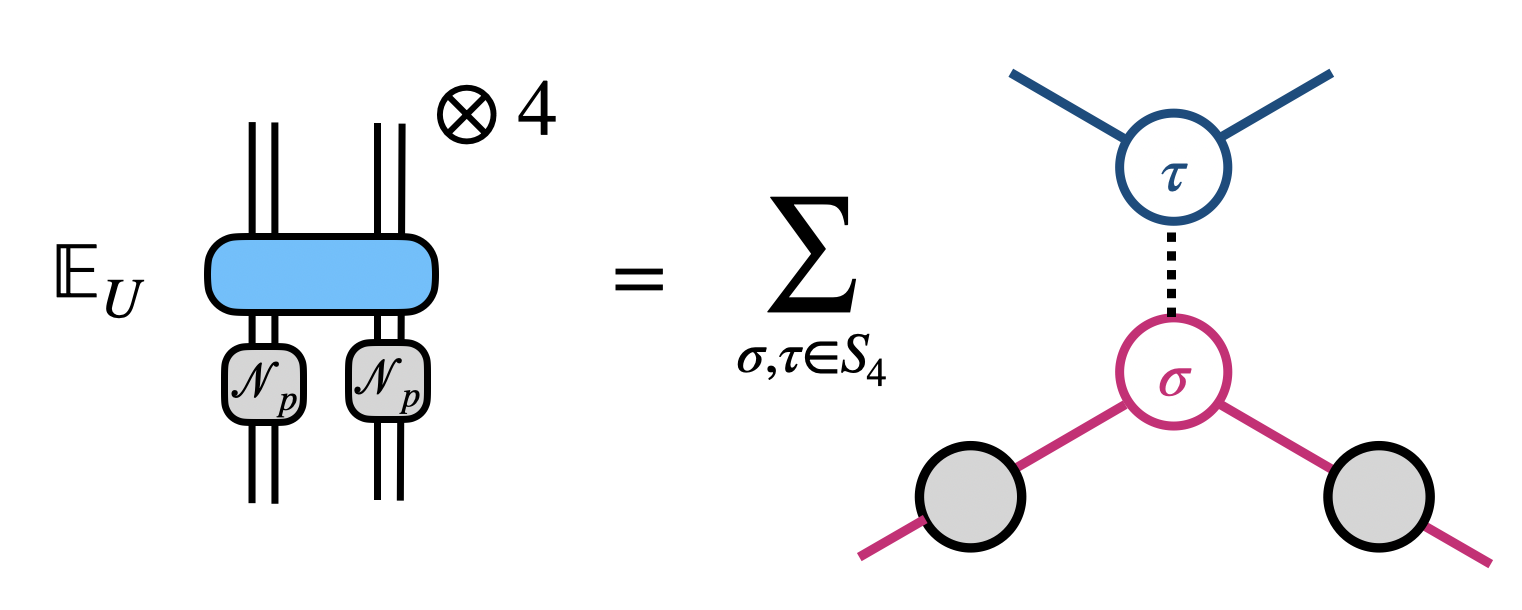}
\end{center}

Where the color coding is the same as in Eq.~\eqref{eq:noisy_haar_average_word}. We use the dashed line to represent $\mathcal{W}(\sigma^{-1}\tau, h^2)$.

The entire circuit is constructed by contracting the above units together, as shown in Fig.~\ref{fig:stat_mech}(a). To simplify the notation, we define a triangle tensor $\textcolor{magenta}{T_{i}^{jk}}$ by integrating out $\tau_i$, shown below.

\begin{equation}\label{eq:triangle_tensor}
    \includegraphics[width=0.9\linewidth]{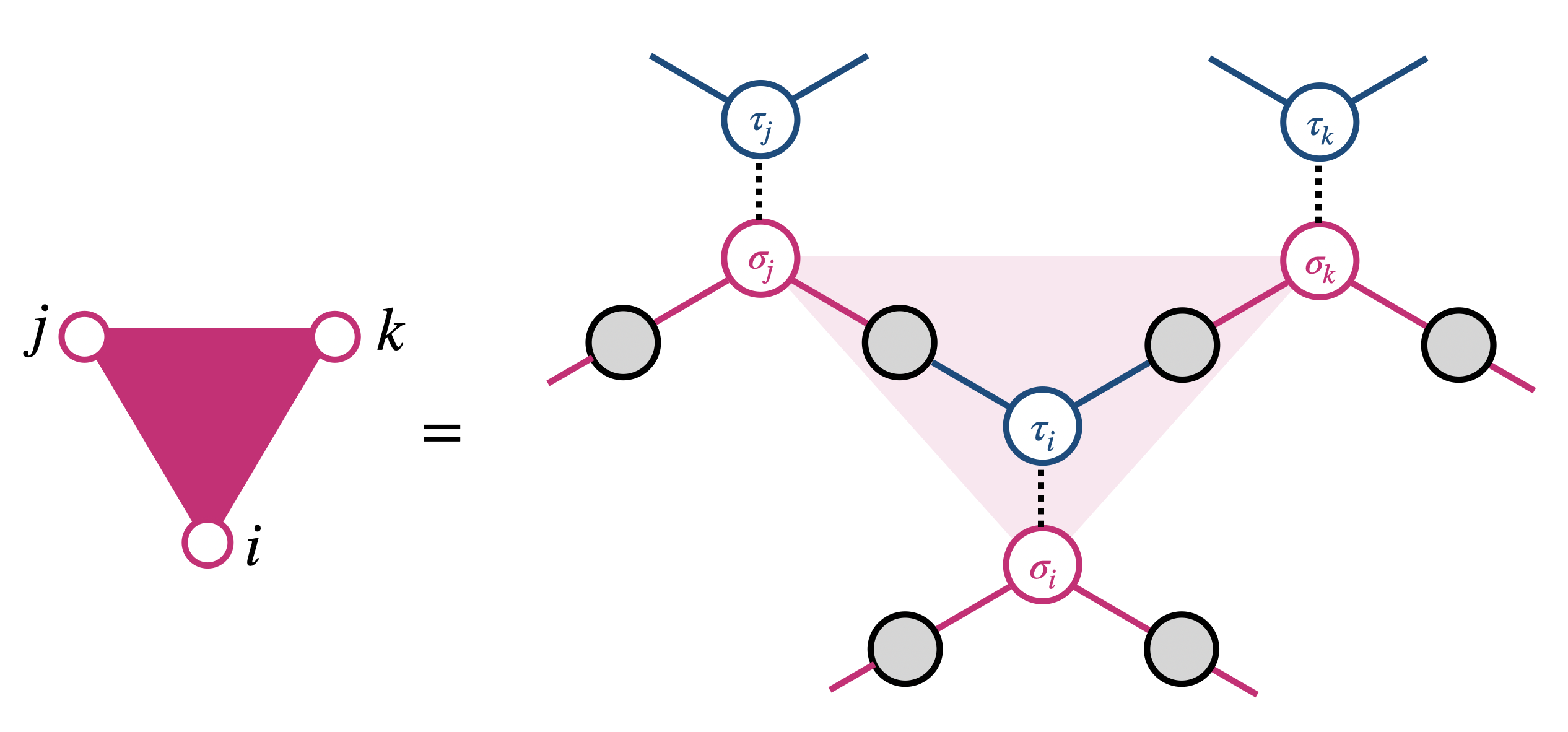}
\end{equation}

The resulting contraction is a summation over the permutations $\sigma$ at each site with a three-body coupling. This summation is reminiscent of the partition function in statistical mechanics. Therefore, from now on we will treat each $\sigma$ as a spin variable with values taken from the symmetric group $S_k$.

\begin{definition}
    \label{def:potts_model} The Potts model is defined as a exponential sum, which we call the partition function, over the spins $\sigma$ on a lattice $L$ with a three-body coupling $\textcolor{magenta}{T_{i}^{jk}}$ in Eq.~\eqref{eq:triangle_tensor}. The boundary conditions will be specified later.
\end{definition}

Note that unlike the usual Potts model, the coupling $\textcolor{magenta}{T_{i}^{jk}}$ can become negative when $k \ge 3$. This is known as the negative sign problem in the literature~\cite{fisher2022random}. In later discussion, we will focus on the large $h$ limit where this problem can be ignored.

We now apply our notation to study the fourth-moment bound in Lemma~\ref{lem:trace_distance_bound}. One can show that the fourth moment bound can be written as a linear combination of the partition function, with different boundary conditions. We use the following convention to order the four copies. We first group into $12$ and $34$. The two norm is computed between the two groups. Within each group, if it evaluates $\tilde{\rho}_{AC|0} \otimes \Tr[\tilde{\rho}_{AC|0}]$, then the first copy ($1$ or $3$) corresponds to $\tilde{\rho}_{AC|0}$ and the second copy ($2$ or $4$) corresponds to $\Tr[\tilde{\rho}_{AC|0}]$. If the group evaluates $\Tr_A[\tilde{\rho}_{AC|0}] \otimes \Tr_C[\tilde{\rho}_{AC|0}]$, then the first copy corresponds to $\Tr_A[\tilde{\rho}_{AC|0}]$ and the second copy corresponds to $\Tr_C[\tilde{\rho}_{AC|0}] $. 

\begin{proposition}
    \label{prop:fourth_moment_partition} The fourth moment bound in Lemma~\ref{lem:trace_distance_bound} can be written as
    \begin{equation}\label{eq:fourth_moment_partition}
        \begin{split}
        &\mathbb{E}_{U \sim \mathrm{Haar}} \\
            &\norm{\tilde{\rho}_{AC|0} \otimes \Tr[\tilde{\rho}_{AC|0}] - \Tr_A[\tilde{\rho}_{AC|0}] \otimes \Tr_C[\tilde{\rho}_{AC|0}]}_2^2  \\
         &= Z_1 - Z_2 - Z_3 + Z_4 
        \end{split}
    \end{equation}
    Where $Z_i$ is the partition function of the 24-state model defined in Definition~\ref{def:potts_model} with the bottom boundary being free, top boundary supported on $B$ being free, and top boundary not supported on $ABC$ being pinned to $e$. The top boundary supported on $A$ and $C$ are pinned to the following:
    \begin{itemize}
        \item $Z_1$: $A$ fixed to $(13)(2)(4)$ and $C$ fixed to $(13)(2)(4)$.
        \item $Z_2$: $A$ fixed to $(13)(2)(4)$ and $C$ fixed to $(14)(2)(3)$.
        \item $Z_3$: $A$ fixed to $(14)(2)(3)$ and $C$ fixed to $(13)(2)(4)$.
        \item $Z_4$: $A$ fixed to $(13)(2)(4)$ and $C$ fixed to $(24)(1)(3)$.
    \end{itemize}
    The previous $\Delta Z$ is defined as $\Delta Z = Z_1 - Z_2 - Z_3 + Z_4$.
\end{proposition}

The boundary condition of $Z_2$ is depicted in Fig.~\ref{fig:stat_mech}(b). In appendix~\ref{app:conn_corr} we show that $\Delta Z$ can be interpreted as the connected correlation of the model. Now we are ready to state the proof of Eq.~\eqref{eq:bar_D_bound} in Theorem~\ref{thm:stat_mech} below.

\begin{proof}[Proof of Eq.(\ref{eq:bar_D_bound}) in Theorem~\ref{thm:stat_mech}]
    We first apply Lemma~\ref{lem:trace_distance_bound} to bound $\bar{D}$ in terms of the fourth moment. Then we apply Proposition~\ref{prop:fourth_moment_partition} to rewrite the fourth moment in terms of the partition function. Finally, we apply Proposition~\ref{prop:connected_correlation} to rewrite the partition function in terms of the connected correlation function. Note that the dimensionality constant can be improved from $h^{3|AC|}$ to $h^{3d (|\partial_A B|+|\partial_C B|)}$, which we discuss in Appendix~\ref{app:improving_dimensionality_constant}.
\end{proof}

\subsubsection{Statistical Mechanics in the Large Hilbert Space Limit}

Next, we bound $\Delta Z$ when $h$ is $\Omega(n)$ with a sufficiently large constant, deriving Eq.~\eqref{eq:bar_D_bound_limit}). We first consider the simplified example where $h \rightarrow \infty$. In this limit, the triangle tensor $\textcolor{magenta}{T_{i}^{jk}}$ in Eq.~\eqref{eq:triangle_tensor} takes the following simplification.

\begin{proposition}
    \label{prop:triangle_tensor_limit} In the limit of large $h$, the triangle tensor $\textcolor{magenta}{T_{i}^{jk}}$ in Eq.~\eqref{eq:triangle_tensor} becomes
\begin{equation}
    \textcolor{magenta}{T_{i}^{jk}} = (1-p) \delta_{i,j,k} + p' \delta_{i,e} \delta_{j,e} \delta_{k,e} + O(\frac{1}{h})
\end{equation}
where $p'=2p - p^2$. As a reminder, $e=(1)(2)(3)(4)$. The above equation can also be represented as a tensor network shown below.

\begin{equation}
    \includegraphics[width=0.9\linewidth]{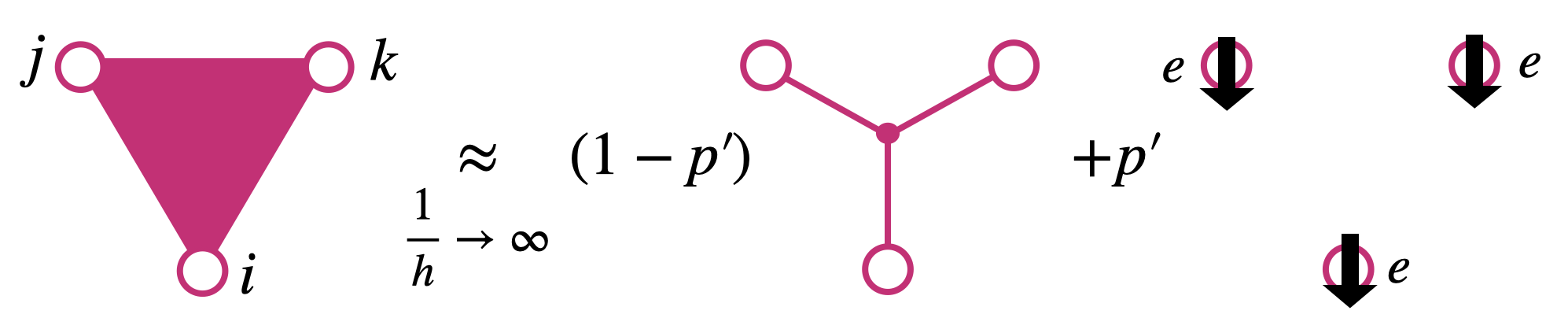}
\end{equation}
\end{proposition}

The proof of the above proposition follows trivially from the infinite Hilbert space limit of the Weingarten function and the rule of inner products between different $\ket{\sigma}$ (See Ref.~\cite{zhou2019emergent, hunter2019unitary}). Therefore, in the large $h$ limit, the Potts model becomes a zero-temperature ferromagnetic model with a pinning field on the configuration $e$. No spins can be misaligned. Thus, the partition function $Z_i$ reduces to a simple analytic formula.

\begin{proposition}
    \label{prop:partition_function_limit}  In the limit when $h \rightarrow \infty$,
    the partition function $Z_i$ can be bounded by
\begin{align}
    Z_1 &= \sum_{\substack{ \sigma(i,j,k) \in \text{triangle} \\ \text{boundary condition 1} }} \textcolor{magenta}{T_{i}^{jk}}  \le (1-p')^{dn} \\
    Z_2 &= Z_3 = Z_4 = 0
\end{align}
In particular, $dn$ is at least $l_{AC} \mathcal{A}$, so $\Delta Z $ decays exponentially in $l_{AC}$.
\end{proposition}
\begin{proof}
    The limit where $h \rightarrow \infty$ follows trivially from the limit of the triangle tensor in Proposition~\ref{prop:triangle_tensor_limit}. In this limit, any misaligned spins will cause the partition function to evaluate to zero. Therefore, $Z_2$, $Z_3$, and $Z_4$ are zero since the boundary conditions contain different spins. For $Z_1$, the boundary spins are aligned, so one only has to compute the probability that no sites in the bulk pins to the identity. This gives $(1-p')^{dn}$. Also note that when $ABC$ do not form the entire system (so that some qubits are traced out), part of the top boundary is pinned to $e$. In this case, $Z_1$ evaluates to zero as well.

\end{proof}

One can see that Proposition~\ref{prop:partition_function_limit} indeed agrees with Theorem~\ref{thm:stat_mech} in the $h \rightarrow \infty$ limit, where $\xi_h = 0$. Computing $\Delta Z$ at a finite $h = \Omega(n)$ is more involved so we present the proof in Appendix~\ref{app:h_correction}. We give an intuition behind the proof. Setting $h$ to some finite value can be considered as a ``low-temperature'' expansion, where $1/h$ plays the role of the temperature. With a finite $h$, $Z_1$ through $Z_4$ are no longer tensor networks with delta tensors: they allow spins to be misaligned. However, the formation of misaligned spin incurs a $O(1/h)$ factor. This can be considered as an energy cost associated with forming a ``domain wall''. As long as $1/h$ is sufficiently small, one can construct a Taylor series resembling the low-temperature expansion in statistical mechanics and show it it converges.

However, we also want the $1/h$ correction terms to decay under the pinning field. To accomplish that, we perform a similar analysis as the $h \rightarrow \infty$ limit. In particular, Proposition~\ref{prop:partition_function_limit} states that pinning any site makes $\Delta Z$ zero. We show that when very few domain walls are formed, then pinning most sites still makes $\Delta Z$ zero. Therefore, pinning still generates a finite correlation length even in the $1/h$ correction terms.



Finally, we discuss how the Markov length $\xi$ is controlled by the noise rate $p$ and the circuit depth $d$. We will focus on the case where $h \rightarrow \infty$ but the same holds for $h = \Omega(n)$. Proposition~\ref{prop:partition_function_limit} shows that $\Delta Z$ decays as $O((1-p')^{l_{AC} \mathcal{A}})$. Therefore, the inverse Markov length is given by $\xi^{-1} = O(\mathcal{A} \log(1-p'))$. In the limit where $p$ is small, we obtain $\xi^{-1} = O(\mathcal{A} p)$. Further, note that $\mathcal{A} \propto d$. Therefore, an increasing noise rate $p$ or an increasing circuit depth $d$ will decrease the Markov length $\xi$. As we will see in the numerical results, this is consistent with the numerical results, even though the Clifford numerics operates in a different regime.


\subsection{Clifford Numerics}

\begin{figure*}
    \includegraphics[width=\linewidth]{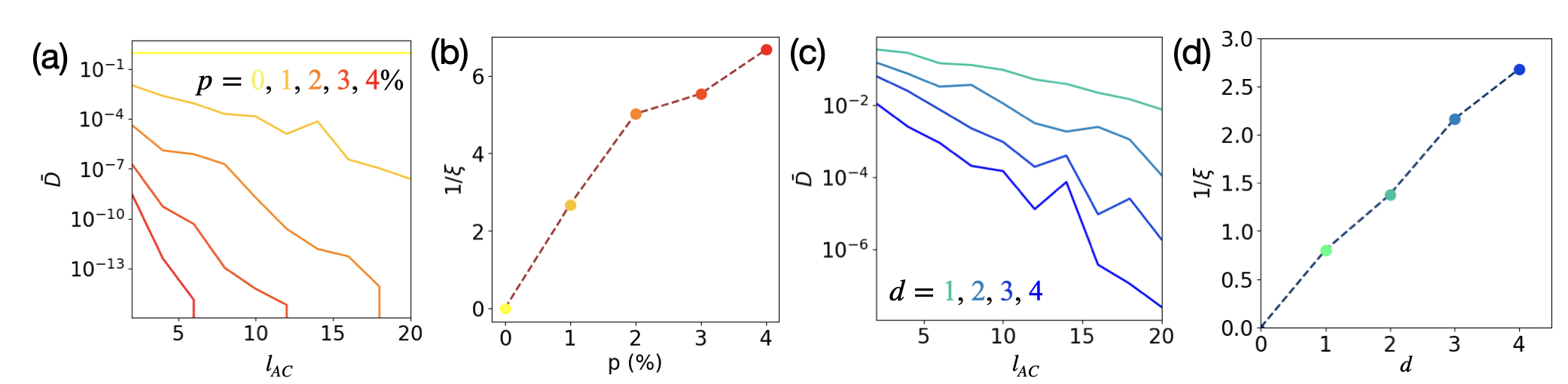}
    \caption{\label{fig:numerics} Numerical results for approximate Markovianity in noisy random quantum circuits. (a) $\bar{D}$ as a function of $l_{AC}$ at different noise rate $p$. (b) Inverse Markov length $1/\xi$ as a function of $p$. (c) $\bar{D}$ as a function of $l_{AC}$ at different depth $d$. (d) $1/\xi$ as a function of $d$.}
\end{figure*}

To complement the analytical bound, we present numerical results on noisy random Clifford circuits. We choose Clifford gates because they can be efficiently simulated classically. In addition, random Clifford circuits do not operate in the regime where the analytical bound applies: the fourth moment property of Clifford gates does not equal the Haar measure. Moreover, we work with qubits $h=2$ so we are far away from the large $h$ limit. Remarkably, we are still able to observe the exponential decay of the trace distance $\bar{D}$ in the distance $l_{AC}$, and the Markov length $\xi$ behaves as predicted by the analytical bound.

The Clifford numerics is based on randomly sampling the noisy random quantum circuits. At each shot, we randomly choose site $i$ and layer $j$ with probability $p$ to completely depolarize the qubit here. $\bar{D}$ is defined analogously except that we average over the random Clifford gates instead of the Haar measure. We compute $\bar{D}$ for different depolarizing rates $p$ and different distances $l_{AC}$. The results are shown in Fig.~\ref{fig:numerics}(a). We observe that $\bar{D}$ decays exponentially in $l_{AC}$ even for very small $p=0.01$. We also observe in Fig.~\ref{fig:numerics}(c) that $\bar{D}$ decays at larger depths, after fixing the depolarizing rate $p$.

To gain a more quantitative understanding of the Markov length, we plot the inverse Markov length $\xi^{-1}$ as a function of $p$ and $l_{AC}$ in Fig.~\ref{fig:numerics}(b,d). We see that $\xi^{-1}$ is proportional to $p$ and $l_{AC}$, consistent with the analytical bound. This is rather surprising since the Clifford numerics is far from the infinite Hilbert space limit as required by the analytical bound. The Clifford numerics indicates that approximate Markovianity holds true for any $O(1)$ noise rate, with a Markov length proportional to $1/(pd)$. We provide the numerical detail in Appendix~\ref{app:clifford_detail}.

Of course, Clifford circuits are always efficiently simulable, regardless of whether they are Markovian. However, many of the correlation properties of random Clifford circuits are empirically similar to those of Haar random circuits. Therefore, establishing the Markov property for Clifford circuits is a suggestive result.

\section{Discussion}\label{sec:discussion}

In this paper, we have presented a quasi-polynomial-time classical sampling algorithm for noisy quantum circuits, under the condition of approximate Markovianity. We rigorously prove that approximate Markovianity holds true in any worst-case circuit when the noise strength is above a constant threshold that only depends on the geometry. We also present analytical and numerical evidence of approximate Markovianity in typical random quantum circuits subject to any constant noise. These results together rule out the possibility of asymptotic quantum advantage in noisy quantum circuits that are not fine-tuned to tolerate weak noise.

We outline several future directions. An immediate question is to prove the approximate Markovianity rigorously in random circuits subject to any constant noise. Many earlier analyses depend on anticoncentration, but since we would like to establish the approximate Markovianity at many depths, new techniques without relying on anticoncentration would be required. We suspect that either one would need to systematically understand the large deviation theory of the constant-depth noisy random circuits, or one needs a new paradigm beyond the moment properties of random circuits.

In addition, breaking the asymptotic quantum advantage does not imply that the classical simulation algorithm beats a quantum computer in terms of clock time. In fact, we believe that without further optimization, our classical simulation algorithm is impractical. Thus, a future direction is to combine our algorithm with existing techniques such as tensor network methods. For example, the computation of the conditional marginal can be performed using tensor network computations, which could be far more efficient than brute-force simulations. Overall, we do not view our algorithm as a competitor to previous works, but as offering a new perspective that can be combined with other techniques to yield more powerful classical simulation algorithms.

While our work considers depolarizing noise, we remark that other noise models such as non-unital noise can drastically change the behavior of noisy quantum circuits. This is because while depolarizing noise, or unital noise in general, never decreases the entropy of the system, non-unital noise can reduce entropy and even purify the system over time. For example, there exists a family of circuits with non-unital noise that can achieve fault-tolerant quantum computation \cite{ben2013quantum}. Moreover, when the non-unital noise is combined with random circuits, it drastically changes the statistical properties of the output distribution \cite{fefferman2024effect}. Therefore, understanding the boundary of quantum advantage in the presence of non-unital noise is an important future direction. 

\subsection{Concurrent Works}

There are several works that will be published concurrently with this paper or after. We mention their results here. Ref.~\cite{SuunSoumik} discusses the classical simulation of typical random quantum circuits under non-unital noise. The algorithm is the same as the bit-by-bit sampling algorithm discussed in our paper. Crucially, while quantum circuits under non-unital noise can generate non-trivial distributions even when they are deep, they show that noisy random circuits are effectively shallow even when the noise is non-unital. They also provide strong numerical evidence that approximate Markovianity holds in these circuits with high probability at any constant noise rates, so that the conditional probability can be computed locally.

In addition, Ref.~\cite{nelson2025limitationsoflocalcirucit} improves their percolation argument in \cite{nelson2024polynomial,nelson2024polynomialtimeclassicalsimulationnoisy} to aribitrary circuits when the circuit depth is $\tilde{\Omega}(1/p)$. While their results have not led to efficient classical sampling algorithms in this regime, they are able to show the exponential decay of Pauli operator expectation values in the size of the support. This suggests a decay of non-local correlations, so the authors conjecture that approximate Markovianity holds in this case.

Finally, Ref.~\cite{wei2025} conducts an extensive numerical study on MIE under noise. They focus on the boundary evolution technique in two-dimensional noisy Clifford circuits, first introduced in \cite{napp2022efficient}. They numerically show that the bond dimension becomes polynomial for any constant noise rate. They also provide an analytic theory based on the clipped gauge and demonstrate consistency with their numerical results. This is consistent with o

\begin{acknowledgments}
We would like to thank Michael Aizenman, Joel Rajakumar, Jon Nelson, Zhiyuan Wei, and Yihui Quek for helpful discussions. Y.Z. and S.G. acknowledge support from NSF QuSEC-TAQS OSI 2326767. S.L. and L.J. acknowledge support from the ARO MURI (W911NF-21-1-0325), AFOSR MURI (FA9550-21-1-0209), NSF (ERC-1941583, OMA-2137642, OSI-2326767, CCF-2312755, OSI-2426975), and Packard Foundation (2020-71479). S.L. is partially supported by the Kwanjeong Educational Foundation.
\end{acknowledgments}

\appendix

\begin{widetext}

\section{Decay of CMI Above a Noise Threshold}\label{app:decay_cmi_high_noise}

Here, we give a formal proof of Theorem~\ref{thm:decay_cmi_high_noise_informal}, which states that the CMI decays exponentially when the noise is above a threshold. While the main text focuses on the circuits on $D$-dimensional lattice, our result applies to more general circuit geometries. To this end, recall that our model of noisy circuit is given as
\begin{equation}
    \mathcal{C} = \left( \otimes_i \mathcal{D}_i \right) \circ \left( \otimes_i \mathcal{N}_{i,d} \right) \circ \mathcal{U}_d \circ \left( \otimes_i \mathcal{N}_{i,d-1} \right) \circ \mathcal{U}_{d-1} \circ \ldots \circ \left( \otimes_i \mathcal{N}_{i,1} \right) \circ \mathcal{U}_1.
\end{equation}
It is convenient to denote the connectivity of this circuit with the interaction graph:

\begin{definition}[Interaction graph of a circuit]
    Given a noisy quantum circuit $\mathcal{C}$ composed of $d$ layers of $k$-qubit gates acting on $n$ qudits with dimension $h$, the \emph{interaction graph} $G = (V, E)$ of this circuit is defined as follows. The vertices $V= \{(i,j) : i \in [n], j \in [d]\}$ correspond to the spacetime locations of the depolarizing channels $\mathcal{N}_{i,j}$, where $i$ labels the qudit and $j$ labels the layer. For two vertices $v_1=(i,j), v_2=(i',j') \in V$, $\{v_1, v_2\} \in E$ if one of the following conditions is satisfied:
    \begin{enumerate}[label=(\roman*)]
        \item $j = j'$ and there exists a unitary gate acting on both qudits $i$ and $i'$ right before the $j$-th noise layer, i.e., in $\mathcal{U}_{j}$.
        \item $j = j'$ and there exists a unitary gate acting on both qudits $i$ and $i'$ right after the $j$-th noise layer, i.e., in $\mathcal{U}_{j-1}$.
        \item $|j - j'| = 1$ and there exists a unitary gate acting on both qudits $i$ and $i'$ in between $j$-th and $j'$-th noise layers, i.e., $\mathcal{U}_{\max(j,j')}$.
    \end{enumerate}
\end{definition}

For example, we show the interaction graph of the one-dimensional brickwork circuits in Fig.~\ref{fig:setup}(b), where each two-qubit gate connects four channels on adjacent sites and in adjacent layers. With this definition of the interaction graph $G$ of the circuit, geometric locality of the qudits can be naturally inherited. For two regions $A, C \subset [n]$, we define the distance $l_{AC}$ between regions $A$ and $C$ as follows.
\begin{equation}
    l_{AC} = \min_{i \in A, i' \in C, j,j' \in [d]} d((i,j), (i',j')),
\end{equation}
where $d((i,j), (i',j'))$ is the length of the shortest path that connects $v_1=(i,j)$ to $v_2=(i',j')$ in the interaction graph $G$. In other words, $l_{AC}$ is the distance between subsets of vertices $\{(i,j): i\in A\}$ and $\{(i',j'): i'\in C\}$. For example, if $\mathcal{C}$ is a brickwork circuit on $D$-dimensional lattice, it is the Manhattan distance between $A$ and $C$. Finally, given a region $L \subset [n]$, consider a subset of vertices $\{(i,j) : i \in L\}$. We denote the boundary of $L$ as $\partial_G L$, defined as
\begin{equation}
    \partial_G L = \{\{v_1, v_2\} \in E : v_1 \in V, v_2 \notin V\},
\end{equation}
i.e., the number of edges that connect vertices in $V$ with those not in $V$.

\begin{figure*}
    \includegraphics[width=0.8\linewidth]{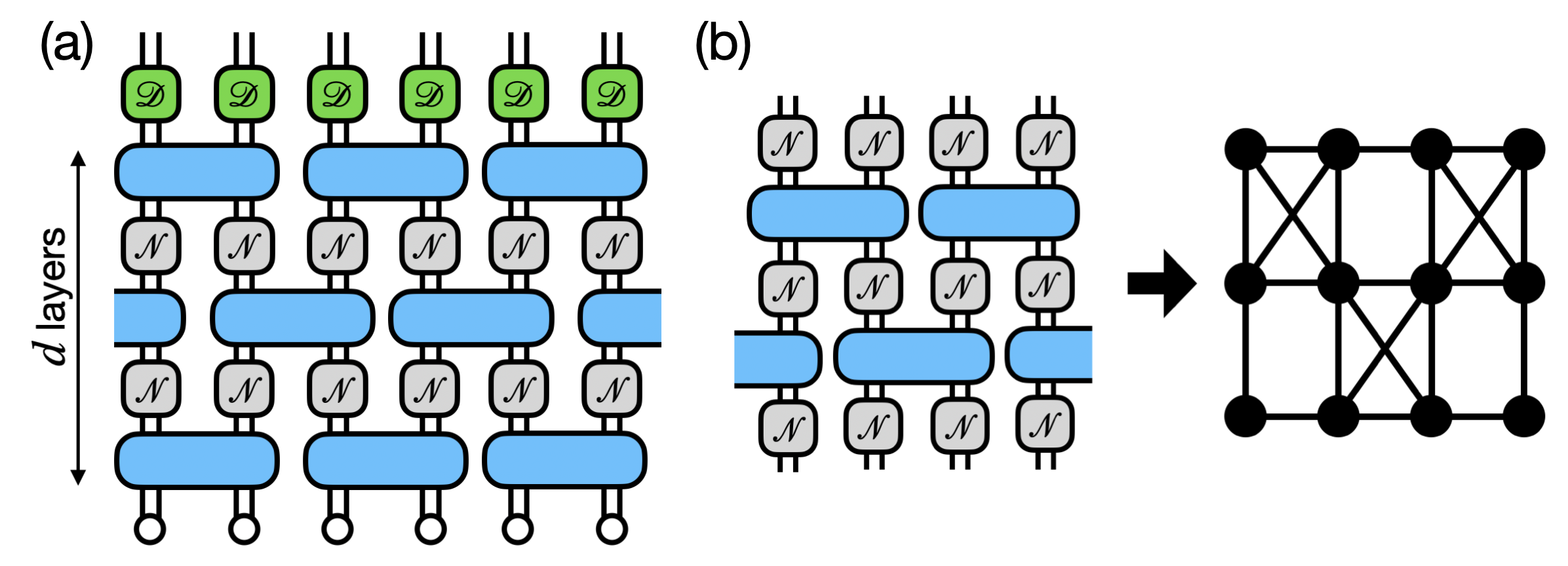}
    \caption{\label{fig:setup} (a) A noisy brickwork circuit in one dimension. (b) The interaction graph associated with the circuit in (a). }
\end{figure*}

With these setups, we restate Theorem~\ref{thm:decay_cmi_high_noise_informal}.

\begin{theorem}[formal restatement of Theorem~\ref{thm:decay_cmi_high_noise_informal}]
\label{thm:decay_cmi_high_noise}
Consider a noisy circuit $\mathcal{C}$ of the form in Eq.~\eqref{eq:noisy_circuit} with noise rate $p = 1-q$, and let $\mathfrak{d}$ be the degree of the interaction graph associated with $\mathcal{C}$. Denoting the measurement distribution as $P$, if $q \le q_c, \forall (i,j)$, where $q_c = \left[ 2h^k(1 + e(\mathfrak{d} - 1))(2e(\mathfrak{d}+1)) \right]^{-1}$, we have
\begin{equation}
    I_P(A:C|B) \le c \cdot \min(|\partial_G A|, |\partial_G C|) \cdot \exp \left(-l_{AC}/\xi\right),
\end{equation}
for some $c=O(1)$ and $\xi = O(1/\log (q/q_{c}))$.
\end{theorem}

In a $D$-dimensional geometrically local, nearest-neighbor brickwork circuit, $l_{AC}$ can be replaced by the Manhattan distance and $|\partial_G A|$, $|\partial_G C|$  can be replaced by $d |\partial A|$, $d |\partial C|$. The degree $\mathfrak{d}$ is also equal to $2Dk$.

We now prove the theorem. As mentioned in Sec.~\ref{sec:decay_cmi}, denoting $\bar{P}_L = P_L \otimes I_{L^c}$ for $L\subset [n]$, we have
\begin{equation}
    I_P(A:C|B) \le \norm{H(A:C|B)},
\end{equation}
where $H(A:C|B) = \log(\bar{P}_{AB}) + \log(\bar{P}_{BC}) -  \log(\bar{P}_{B}) -  \log(\bar{P}_{ABC})$ (see Propositions~\ref{prop:cmi_matrix_main}). To proceed, it is convenient to choose a different normalization,
\begin{equation}
    \tilde{P}_L = h^{|L|} P_L \otimes I_{L^c},
\end{equation}
so that $\tilde{P}_L$ is the diagonal density matrix after applying depolarizing channels with $q=0$ to all sites not in $L$, i.e.,
\begin{equation}
    \tilde{P}_L = \left(\otimes_{i \in L^c} \mathcal{N}^{q=0}_{i}\right)\rho.
\end{equation}
We have the same formula for $H(A:C|B)$ with $\tilde{P}_L$ replacing $\bar{P}_L$,
\begin{equation}
    H(A:C|B) = \log(\tilde{P}_{AB}) + \log(\tilde{P}_{BC}) -  \log(\tilde{P}_{B}) -  \log(\tilde{P}_{ABC}),
\end{equation}
since $\tilde{P}_L = h^{|L|} \bar{P}_L$ and the additional factor $h^{|L|}$ cancels out after taking the logarithm. Therefore, we will work with $\tilde{P}_L$. With this notation, the remaining task is to bound $\norm{H(A:C|B)}$ using the generalized cluster expansion technique.

\subsection{Cluster Expansion}

As mentioned in Sec.~\ref{sec:decay_cmi}, we will perform a multi-variate Taylor expansion of $H(A:C|B)$ with respect to $q_{i,j}$ for all $(i,j)$,

\begin{multline}\label{eq:taylor_expansion}
    H(A:C|B) = \left( \sum_{i,j} q_{i,j} \frac{\partial}{\partial q_{i,j}} \right) H(A:C|B)
    + \left( \sum_{i,j} \frac{q_{i,j}^2}{2} \frac{\partial^2}{\partial q_{i,j}^2} + \sum_{(i,j) \neq (i',j')} q_{i,j}q_{i',j'} \frac{\partial^2}{\partial q_{i,j} \partial q_{i',j'}} \right)H(A:C|B) + \cdots.
\end{multline}
To write the expansion in a more compact way, we define \emph{clusters} and related notations.

\begin{definition}[Cluster]
    Given a circuit with the interaction graph $G = (V,E)$, a \emph{cluster} $\mathbf{W}$ is a multi-set of $q_{i,j}$, which generalizes set allowing each element appearing multiple times. Cluster is represented by set of ordered pairs, $\mathbf{W} = \{(q_{i,j}, \mu_{i,j}): (i,j) \in V\}$, where $\mu_{i,j} \in \mathbb{Z}_{\ge 0}$ is the multiplicity of $q_{i,j}$ in $\mathbf{W}$. We also introduce following notations associated with a cluster $\mathbf{W}$:
    \begin{enumerate}[label=(\roman*)]
        \item We say two clusters $\mathbf{W}_1 = \{(q_{i,j}, \mu_{i,j}): (i,j) \in V\}$ and $\mathbf{W}_2 = \{(q_{i,j}, \mu'_{i,j}): (i,j) \in V\}$ are \emph{disjoint} if for all $(i,j) \in V$, at most one of $\mu_{i,j}$ and $\mu'_{i,j}$ is non-zero.
        \item The \emph{union} of two clusters $\mathbf{W}_1 = \{(q_{i,j}, \mu_{i,j}): (i,j) \in V\}$ and $\mathbf{W}_2 = \{(q_{i,j}, \mu'_{i,j}): (i,j) \in V\}$ is defined as $\mathbf{W}_1 \cup \mathbf{W}_2 = \{(q_{i,j}, \max(\mu_{i,j}, \mu'_{i,j})): (i,j) \in V\}$. If $\mathbf{W}_1$ and $\mathbf{W}_2$ are disjoint, we denote the union as $\mathbf{W}_1 \sqcup \mathbf{W}_2$.
        \item The \emph{weight} $|\mathbf{W}|$ of the cluster $\mathbf{W}$ is defined as $|\mathbf{W}| = \sum_{(i,j) \in V} \mu_{i,j}$, and we say $\mathbf{W}$ is \emph{empty} if $|\mathbf{W}| = 0$.
        \item The \emph{support} of the cluster $\mathbf{W}$ is defined as ${\rm supp}(\mathbf{W}) = \{(i,j) \in V : \mu_{i,j} > 0\}$.
        \item The \emph{cluster derivative} $D_{\mathbf{W}}$ is defined as $D_{\mathbf{W}} = \prod_{(i,j) \in V} \left(\frac{\partial}{\partial q_{i,j}}\right)^{\mu_{i,j}}\Bigr|_{q_{i,j}=0}$.
        \item The product of all $q_{i,j}$ in $\mathbf{W}$, including their multiplicity, is defined as $q_\mathbf{W} = \prod_{(i,j) \in V} (q_{i,j})^{\mu_{i,j}}$.
        \item We also denote $\mathbf{W}! = \prod_{(i,j) \in V} \mu_{i,j}!$.
    \end{enumerate}
\end{definition}

Using these notations, we can write the multi-variate Taylor expansion of Eq.~\eqref{eq:taylor_expansion} as
\begin{equation}\label{eq:cluster_expansion_W}
    H(A:C|B) = \sum_{\mathbf{W}} \frac{q_\mathbf{W}}{\mathbf{W}!} D_{\mathbf{W}} H(A:C|B).
\end{equation}
One simple observation is that since every $\mathcal{N}_{i,j}$ is linear in $q_{i,j}$, so when any $\mu_{i,j} >1$, $D_{\mathbf{W}} \tilde{P}_L = 0$:

\begin{proposition}\label{prop:linear_cluster}
    When $\mathbf{W}$ has any $\mu_{i,j} > 1$, $D_{\mathbf{W}} \tilde{P}_L=0$ for all $L \subset [n]$.
\end{proposition}

Therefore, when calculating $D_{\mathbf{W}} \tilde{P}_L$, we only need to consider clusters with $\mu_{i,j} \in \{0,1\}$ for all $(i,j)$. When $\mu_{i,j} = 1$, we replace the channel $\mathcal{N}_{i,j}$ with the following map $\Theta$, defined as the derivative of the depolarizing channel $\mathcal{N}_{i,j}$ with respect to $q_{i,j}$. 
\begin{equation} \label{eq:theta}
    \Theta[\rho] = \frac{\partial}{\partial q} \mathcal{N}[\rho] = \rho  - \frac{1}{h} I \cdot \Tr[\rho]
\end{equation}
On the other hand, if $\mu_{i,j}=0$, we replace the channel with the complete depolarizing channel. See examples in Fig.~\ref{fig:disconnection}.

Another simple observation is that $\mathbf{W}\tilde{P}_L \propto I$ if $\mathbf{W}$ contains no vertices in the last layer of the circuit and in $L$:
\begin{proposition} 
    \label{prop:cluster_support}
    Given a cluster $\mathbf{W}$, if all $\mu_{i,d} = 0$ for $i \in L$, then $D_{\mathbf{W}}\tilde{\rho}_L \propto I$.
\end{proposition}

Instead of taking the cluster expansion of $\tilde{P}_L$, what is more relevant is the cluster expansion of $H(A:C|B)$, which amounts to taking the cluster derivative of $\log(\tilde{P}_L)$.
\begin{equation}
    \log(\tilde{P}_L) = \sum_{\mathbf{W}} \frac{q_\mathbf{W}}{\mathbf{W}!} D_{\mathbf{W}} \log(\tilde{P}_L)
\end{equation}
Note that the cluster derivative of $\log(\tilde{P}_L)$ is not zero when $\mu_{i,j} > 1$, as the logarithm is not linear.

\subsection{Connected Clusters}

Significant challenge for bounding $\norm{H(A:C|B)}$ is that while the contribution from each cluster $\mathbf{W}$ is exponentially suppressed by its weight, the number of such clusters having that weight rapidly grows in $|\mathbf{W}|$. Here, we provide key observations that help to tame the growth of the number of clusters. To this end, we define \emph{disconnected} and \emph{connected} clusters.

\begin{definition}
    Let $G=(V, E)$ be the interaction graph of the circuit. A cluster $\mathbf{W}$ is called \emph{disconnected} if there exist two disjoint non-empty clusters $\mathbf{W}_1=\{(q_{i,j}, \mu_{i,j}) : (i,j) \in V\}$ and $\mathbf{W}_2=\{(q_{i,j}, \mu'_{i,j}) : (i,j) \in V\}$ such that $\mathbf{W} = \mathbf{W}_1 \sqcup \mathbf{W}_2$, and there is no edge in $E$ that connects ${\rm supp}(\mathbf{W}_1)$ and ${\rm supp}(\mathbf{W}_2)$. Otherwise, we say $\mathbf{W}$ is \emph{connected}.
\end{definition}

An example of a disconnected cluster is shown in Fig.~\ref{fig:disconnection}(c). We show that the cluster derivative $D_{\mathbf{W}} \tilde{P}_L$ factorizes when $\mathbf{W}$ is disconnected. 
\begin{proposition}
    \label{prop:disconnected_factorize} With the interaction graph of the circuit $G$, let $\mathbf{W}$ be a disconnected cluster and $\mathbf{W} = \mathbf{W}_1 \sqcup \mathbf{W}_2$ with ${\rm supp}(\mathbf{W}_1)$ and ${\rm supp}(\mathbf{W}_2)$ are disconnected in $G$. Then, we have
    \begin{equation}
        D_{\mathbf{W}} \tilde{P}_L = \left(D_{\mathbf{W}_1} \tilde{P}_L\right) \cdot \left(D_{\mathbf{W}_2} \tilde{P}_L\right),
    \end{equation}
    for all $L \subset [n]$.
\end{proposition}
\begin{proof}
    First, since $\tilde{P}_L$ is a diagonal matrix, $D_{\mathbf{W}_1} \tilde{P}_L$ and $D_{\mathbf{W}_2} \tilde{P}_L$ trivally commute. Now, since ${\rm supp}(\mathbf{W}_1)$ and ${\rm supp}(\mathbf{W}_2)$ are disconnected in $G$, $\tilde{P}_L \Bigr|_{q_{i,j}=0, \forall q_{i,j} \notin \mathbf{W}}$ is in the form of
    \begin{equation}
        \tilde{P}_L \Bigr|_{q_{i,j}=0, \forall q_{i,j} \notin \mathbf{W}} = \tilde{\sigma}_{{\rm supp}'(\mathbf{W}_1)} \otimes \tilde{\sigma}_{{\rm supp}'(\mathbf{W}_2)} \otimes I_{\rm rest.},
    \end{equation}
    where $\tilde{\sigma}_{{\rm supp}'(\mathbf{W}_1)}$ and $\tilde{\sigma}_{{\rm supp}'(\mathbf{W}_2)}$ arise from
    \begin{align}
        \tilde{P}_L \Bigr|_{q_{i,j}=0, \forall q_{i,j} \notin \mathbf{W}_1} &= \tilde{\sigma}_{{\rm supp}'(\mathbf{W}_1)} \otimes I_{({\rm supp}'(\mathbf{W}_1))^c}, \\
        \tilde{P}_L \Bigr|_{q_{i,j}=0, \forall q_{i,j} \notin \mathbf{W}_2} &= \tilde{\sigma}_{{\rm supp}'(\mathbf{W}_2)} \otimes I_{({\rm supp}'(\mathbf{W}_2))^c}.
    \end{align}
    Therefore,
    \begin{equation}
        \tilde{P}_L \Bigr|_{q_{i,j}=0, \forall q_{i,j} \notin \mathbf{W}}
        = \left(\tilde{P}_L \Bigr|_{q_{i,j}=0, \forall q_{i,j} \notin \mathbf{W}_1}\right) \cdot \left(\tilde{P}_L \Bigr|_{q_{i,j}=0, \forall q_{i,j} \notin \mathbf{W}_2}\right).
    \end{equation}
    For example, see Fig.~\ref{fig:disconnection}(c). Since $D_{\mathbf{W}} = D_{\mathbf{W}_1}D_{\mathbf{W}_2}$, we have
    \begin{align}
        D_{\mathbf{W}} \tilde{P}_L &= D_{\mathbf{W}} \tilde{P}_L\Bigr|_{q_{i,j}=0, \forall q_{i,j} \notin \mathbf{W}}\\
        &= D_{\mathbf{W}} \left(\left(\tilde{P}_L \Bigr|_{q_{i,j}=0, \forall q_{i,j} \notin \mathbf{W}_1}\right) \cdot \left(\tilde{P}_L \Bigr|_{q_{i,j}=0, \forall q_{i,j} \notin \mathbf{W}_2}\right)\right) \\
        &= \left(D_{\mathbf{W}_1} \tilde{P}_L\right) \cdot \left(D_{\mathbf{W}_2} \tilde{P}_L\right).
    \end{align}
\end{proof}

With these properties, we can show that in the expansion of Eq.~\eqref{eq:cluster_expansion_W}, the only terms that contribute to $H(A:C|B)$ are connected clusters that connect $A$ and $C$:

\begin{lemma}[formal restatement of Proposition~\ref{prop:distance_bound}]
    \label{lem:cluster_derivative_hamiltonian}
    Consider a cluster $\mathbf{W}$. If one of the following conditions is satisfied, then $D_{\mathbf{W}} H(A:C|B) = 0$.
    \begin{enumerate}[label=(\roman*)]
        \item $\mathbf{W}$ is disconnected
        \item $\mathbf{W}$ does not contain two vertices $(i,j)$ and $(i',j')$ where $i \in A$ and $i' \in C$.
    \end{enumerate}
    In particular, if $|W| < l_{A,C}$, $D_{\mathbf{W}} H(A:C|B) = 0$.
\end{lemma}

\begin{proof}
    First, suppose $\mathbf{W}$ is disconnected and $\mathbf{W} = \mathbf{W}_1 \sqcup \mathbf{W}_2$ for some non-empty clusters $\mathbf{W}_1$ and $\mathbf{W}_2$, where ${\rm supp}(\mathbf{W}_1)$ and ${\rm supp}(\mathbf{W}_2)$ are disconnected in $G$. As pointed out in the proof of Proposition~\ref{prop:disconnected_factorize}, we have
    \begin{equation}
        \tilde{P}_L \Bigr|_{q_{i,j}=0, \forall q_{i,j} \notin \mathbf{W}} = \left(\tilde{P}_L \Bigr|_{q_{i,j}=0, \forall q_{i,j} \notin \mathbf{W}_1}\right) \cdot \left(\tilde{P}_L \Bigr|_{q_{i,j}=0, \forall q_{i,j} \notin \mathbf{W}_2}\right),
    \end{equation}
    for all $L \subset [n]$. Therefore, we have
    \begin{equation}
        \log(\tilde{P}_L)\Bigr|_{q_{i,j}=0, \forall q_{i,j} \notin \mathbf{W}} = \log(\tilde{P}_L)\Bigr|_{q_{i,j}=0, \forall q_{i,j} \notin \mathbf{W}_1} + \log(\tilde{P}_L)\Bigr|_{q_{i,j}=0, \forall q_{i,j} \notin \mathbf{W}_2},
    \end{equation}
    since the first and second terms commute. Note that since both $\mathbf{W}_1$ and $\mathbf{W}_2$ are non-empty and disjoint, we have
    \begin{align}
        D_{\mathbf{W}_2}\left(\log(\tilde{P}_L)\Bigr|_{q_{i,j}=0, \forall q_{i,j} \notin \mathbf{W}_1}\right) &= 0,\\
        D_{\mathbf{W}_1}\left(\log(\tilde{P}_L)\Bigr|_{q_{i,j}=0, \forall q_{i,j} \notin \mathbf{W}_2}\right) &= 0.
    \end{align}
    Therefore, we have $D_{\mathbf{W}} \log(\tilde{\rho}_L) = 0$ for all $L \subset [n]$, and thus $D_{\mathbf{W}} H(A:C|B) = 0$.

    We now discuss the second condition. Without loss of generality, suppose $\mathbf{W}$ does not contain vertices $(i',j')$ where $i' \in C$. First note that
    \begin{equation}
        \tilde{P}_{ABC} \Bigr|_{q_{i,j}=0, \forall q_{i,j} \notin \mathbf{W}} = \tilde{P}_{AB} \Bigr|_{q_{i,j}=0, \forall q_{i,j} \notin \mathbf{W}},
    \end{equation}
    since $\mathbf{W}$ does not contain any vertices in $C$, see Fig.~\ref{fig:AC_conn}. Therefore, we have
    \begin{equation}
        D_{\mathbf{W}} \log\left(\tilde{P}_{ABC}\right) = D_{\mathbf{W}} \log\left(\tilde{P}_{AB}\right).
    \end{equation}
    Similarly, we can show that
    \begin{equation}
        D_{\mathbf{W}} \log\left(\tilde{P}_{ABC}\right) = D_{\mathbf{W}} \log\left(\tilde{P}_{AB}\right).
    \end{equation}
    where $\rho'_{B}$ is supported on $B$ only. On the other hand, we can also set all $q_{i,j} \notin \mathbf{W}$ to zero in $\tilde{\rho}_{B}$, and we use the fact that $\tilde{\rho}_{B} = \otimes_{i \in C} \mathcal{N}^{p=1}_{i}[\tilde{\rho}_{BC}]$.
    \begin{align}
        \tilde{\rho}_{B} \Bigr|_{q_{i,j}=0, \forall q_{i,j} \notin \mathbf{W}} &= \otimes_{i \in C} \mathcal{N}^{p=1}_{i}[\tilde{\rho}_{BC}]\Bigr|_{q_{i,j}=0, \forall q_{i,j} \notin \mathbf{W}} \\
        &= \otimes_{i \in C} \mathcal{N}^{p=1}_{i}[\rho'_{B}] = \rho'_{B}
    \end{align}
    Finally, all four terms in $D_{\mathbf{W}} H(A:C|B)$ cancel out,
    \begin{align}
        D_{\mathbf{W}} H(A:C|B) &= D_{\mathbf{W}} \left( \log(\tilde{P}_{AB}) -  \log(\tilde{P}_{ABC}) + \log(\tilde{P}_{BC}) -  \log(\tilde{P}_{B})  \right) = 0.
    \end{align}

    Finally, for $W$ being connected as well as containing vertices in $A$ and $C$, $|W|$ needs to be at least $l_{A,C}+1$. Therefore, whenever $|W| \le l_{A,C}$, $D_{\mathbf{W}}H(A:C|B) = 0$.
\end{proof}

Therefore, the cluster expansion of $H(A:C|B)$ only contains connected clusters that connect $A$ and $C$. Specifically, let us denote $\mathcal{G}^{AC}_m$ as the set of all weight-$m$ connected clusters that contain at least two vertices $(i,j)$ and $(i',j')$ where $i \in A$ and $i' \in C$. Then, we have
\begin{equation}
    H(A:C|B) = \sum_{m=0}^{\infty} \sum_{\mathbf{W} \in \mathcal{G}^{AC}_m} \frac{q_\mathbf{W}}{\mathbf{W}!} D_{\mathbf{W}} H(A:C|B)
\end{equation}
This restriction on the clusters that contribute to $H(A:C|B)$ significantly reduces the number of clusters we need to consider. In particular, the following lemma shows that the number of relevant clusters grows at most exponentially in the weight:

\begin{lemma}[Proposition~3.6 of~\cite{haah2022optimal}] \label{lem:connected_cluster_bound}
    Let $G=(V,E)$ be the interaction graph of the circuit. The number of connected clusters $\mathbf{W}$ supported on one site $(i,j)$ with $|\mathbf{W}|=m$ is upper-bounded by $e\mathfrak{d}(1 + e(\mathfrak{d} - 1))^{m-1}$, where $\mathfrak{d}$ denotes the degree of $G$. 
\end{lemma}

Therefore, if we can establish that the operator norm of the cluster derivative of $H(A:C|B)$ decays exponentially, the cluster expansion will converge and we can bound the CMI.

\begin{figure*}
    \includegraphics[width=0.9\linewidth]{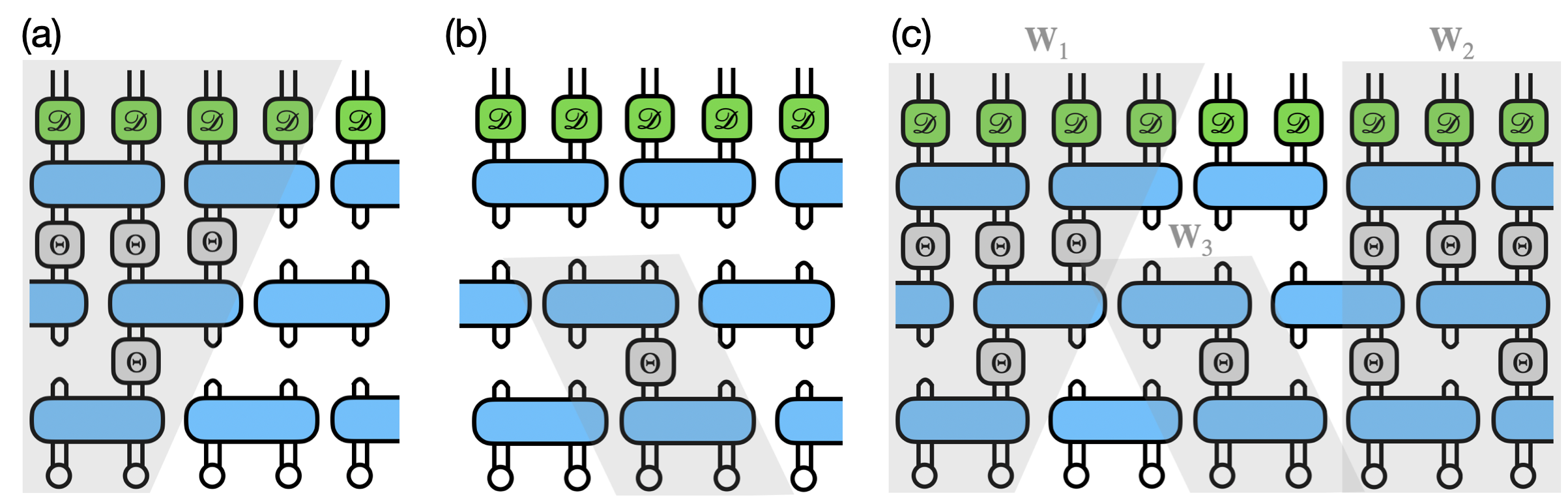}
    \caption{\label{fig:disconnection} (a) A cluster supported on the last layer of the circuit. It corresponds to a diagonal matrix that may not be proportional to the identity. (b) A cluster supported deep in the circuit. It corresponds to a constant multiples of the the identity. (c) Cluster derivative $D_{\mathbf{W}} \rho_L$ factorizes when $\mathbf{W}$ is disconnected.}
\end{figure*}

\begin{figure*}
    \includegraphics[width=0.6\linewidth]{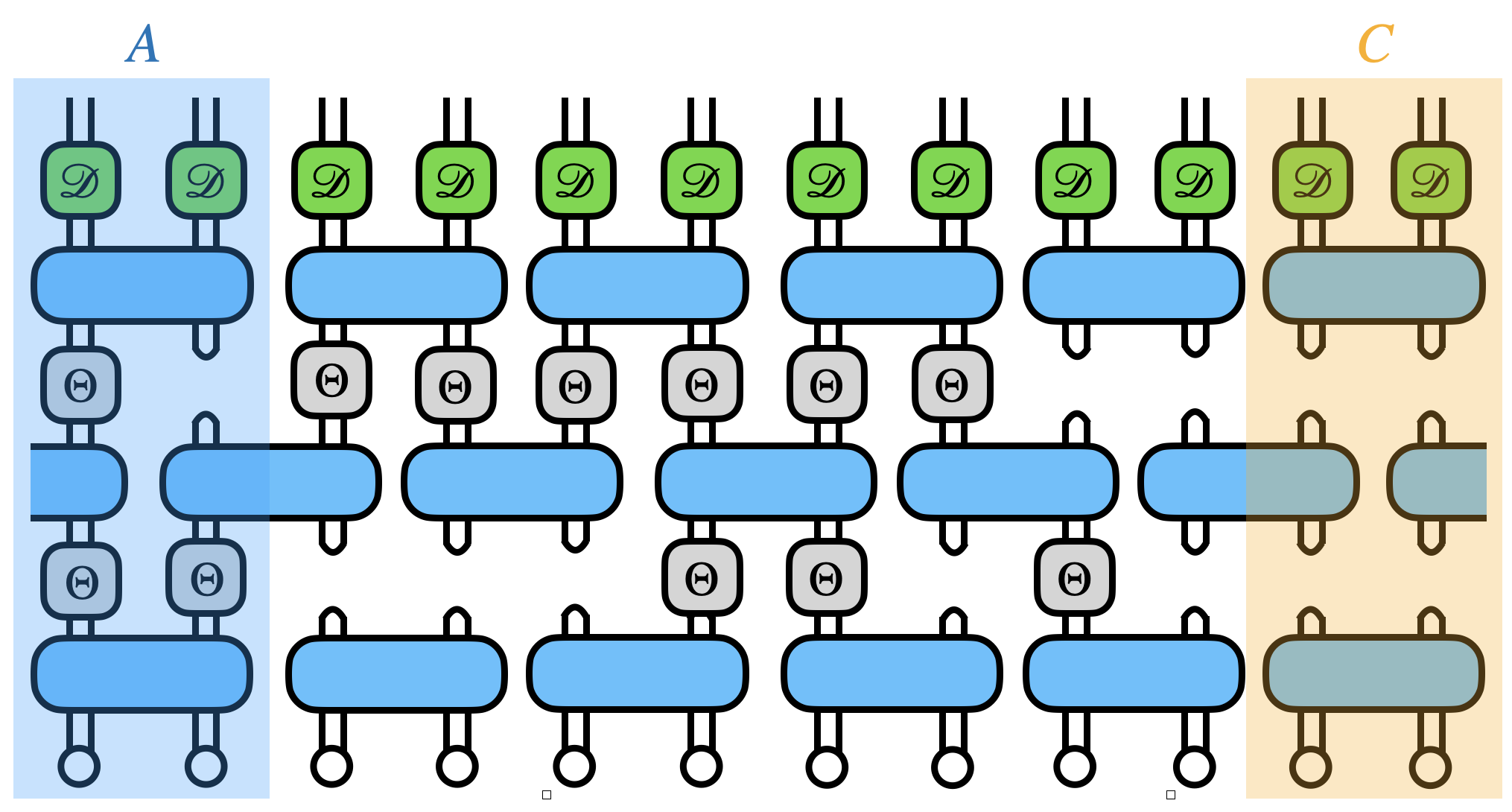}
    \caption{\label{fig:AC_conn} When a cluster does not connect to $C$, the density matrix becomes the maximally mixed state on $C$. }
\end{figure*}

\subsection{Proof of Theorem~\ref{thm:decay_cmi_high_noise}}

We now provide the proof of Theorem~\ref{thm:decay_cmi_high_noise}. The last ingredient we need is an upper bound on the operator norm of the individual cluster derivative $\norm{D_{\mathbf{W}}H(A:C|B)}$. The following lemma, which we prove later in Appendix~\ref{app:cluster_deriv_bound}, provides such an upper bound.

\begin{lemma}
    \label{lem:cluster_deriv_upper}
        Let $G=(V,E)$ be the interaction graph of the circuit. Given any cluster $\mathbf{W}$ with order $|\mathbf{W}|=m$, we have
    \begin{align}
        \frac{1}{\mathbf{W}!}\norm{D_{\mathbf{W}} \log(\tilde{\rho}_L)} \le (2h^k)^{m} (2e(\mathfrak{d}+1))^{m+1}
    \end{align}
    where $\mathfrak{d}$ is the degree of $G$.
\end{lemma}

With this lemma, we now have all the ingredients to prove Theorem~\ref{thm:decay_cmi_high_noise}.

\begin{proof}[Proof of Theorem~\ref{thm:decay_cmi_high_noise}]
    To start with, recall
    \begin{equation}
        I_P(A:C|B) \le \norm{H(A:C|B)},
    \end{equation}
    where
    \begin{equation}
        H(A:C|B) = \sum_{m=0}^{\infty} \sum_{\mathbf{W} \in \mathcal{G}^{AC}_m} \frac{q_\mathbf{W}}{\mathbf{W}!} D_{\mathbf{W}} H(A:C|B).
    \end{equation}
    Since $\mathbf{W} \in \mathcal{G}^{AC}_m$ connects $A$ and $C$, the weight of $\mathbf{W}$ is at least the distance between $A$ and $C$, i.e., $|\mathbf{W}| \ge l_{AC}$. Therefore, $\mathcal{G}^{AC}_m = \emptyset$ when $m < l_{AC}$. Thus, we can rewrite
    \begin{equation}
        H(A:C|B) = \sum_{m=l_{AC}}^{\infty} \sum_{\mathbf{W} \in \mathcal{G}^{AC}_m} \frac{q_\mathbf{W}}{\mathbf{W}!} D_{\mathbf{W}} H(A:C|B),
    \end{equation}
    and by the triangle inequality,
    \begin{equation}
        \norm{H(A:C|B)} \le \sum_{\mathbf{W}} \frac{q_{\mathbf{W}}}{\mathbf{W}!}\norm{D_\mathbf{W} H(A:C|B)}.
    \end{equation}
    Applying Lemma~\ref{lem:cluster_deriv_upper}, we further have
    \begin{equation}\label{eq:main_thm_ineq_step}
        \norm{H(A:C|B)} \le 4 \sum_{m=l_{AC}}^{\infty} |\mathcal{G}^{AC}_m| q^m (2h^k)^{m} (2e(\mathfrak{d}+1))^{m+1}
    \end{equation}
    where we used the fact that $H(A:C|B)$ contains four terms.

    For $|\mathcal{G}^{AC}_m|$, note that for a connected cluster to connect $A$ and $C$, it has to be supported on $\partial_G A$ and $\partial_G C$. Therefore, the number of clusters $|\mathcal{G}^{AC}_m|$ in $\mathcal{G}^{AC}_m$ is upper-bounded by
    \begin{equation}
        |\mathcal{G}^{AC}_m| \le \min(|\partial_G A|, |\partial_G C|) e\mathfrak{d}(1 + e(\mathfrak{d} - 1))^{m-1}
    \end{equation}
    where we used Lemma~\ref{lem:connected_cluster_bound}. Combining this with Eq.~\eqref{eq:main_thm_ineq_step}, we have
    \begin{equation}
        \norm{H(A:C|B)} \le 4 \min(|\partial_G A|, |\partial_G C|) \sum_{m=l_{AC}}^{\infty} q^m e\mathfrak{d}(1 + e(\mathfrak{d} - 1))^{m-1} (2h^k)^{m} (2e(\mathfrak{d}+1))^{m+1}.
    \end{equation}
    Let $q_c = \left[2 h(1 + e(\mathfrak{d} - 1))(2e(\mathfrak{d}+1)) \right]^{-1}$. The above series converges absolutely when $q \le q_c$. We pack the constant coefficients into $c'$ to get
    \begin{align}
        \norm{H(A:C|B)} &\le c' \min(|\partial_G A|, |\partial_G C|) \sum_{m=l_{AC}}^{\infty} \left(\frac{q}{q_c} \right)^m \\
        &= \frac{c'}{1- q/q_c} \min(|\partial_G A|, |\partial_G C|)\left(\frac{q}{q_c} \right)^{l_{AC}},
    \end{align}
    in which we arrive at the desired result.
\end{proof}

\subsection{\label{app:cluster_deriv_bound}Proof of Lemma~\ref{lem:cluster_deriv_upper}}

Finally, we prove Lemma~\ref{lem:cluster_deriv_upper}, which provides an upper bound on the operator norm of the cluster derivative of $\log(\tilde{\rho}_L)$. We first quote a few useful facts regarding how channels modify the operator norm.

\begin{proposition}[Theorem II.4 of~\cite{perez2006contractivity}]\label{prop:unital_channel_norm}
    For a finite dimensional Hilbert space $\mathcal{H}$, let $\mathcal{N}:\mathcal{L}(\mathcal{H}) \rightarrow \mathcal{L}(\mathcal{H})$ be a unital channel. Then for any matrix $M \in \mathcal{L}(\mathcal{H})$ we have,
    \begin{equation}
        \norm{\mathcal{N}[M]} \le \norm{M}.
    \end{equation} 
\end{proposition}

\begin{proposition}\label{prop:theta_norm}
    Consider the linear map $\Theta$ defined in Eq.~\eqref{eq:theta}. Then for any matrix $M \in \mathcal{L}(\mathcal{H})$ with an $n$-qubit Hilbert space $\mathcal{H}$, we have
    \begin{equation}
        \norm{\Theta[M]} \le 2 \norm{M}.
    \end{equation}
\end{proposition}
\begin{proof}
    We write the depolarizing channel as follows.
    \begin{equation}
        \mathcal{N}[M] = q M + (1-q)\cdot \frac{1}{h^2}\sum_{i=1}^{h^2} A_i M A_i^\dag,
    \end{equation}
    where $A_i$ are elements of the Weyl-Heisenberg group (a $d$-dimensional generalization of the Pauli group). We take the derivative with respect to $q$ to get
    \begin{equation}
       \Theta[M] = M - \frac{1}{h^2}\sum_{i=1}^{h^2} A_i M A_i^\dag.
    \end{equation}
    Using the triangle inequality and convexity of the operator norm, we have
    \begin{align}
        \norm{\Theta[M]} &\le \norm{M} + \norm{\frac{1}{h^2}\sum_{i=1}^{h^2} A_i M A_i^\dag} \\
        &\le \norm{M} + \frac{1}{h^2}\sum_{i=1}^{h^2} \norm{A_i M A_i^\dag}.
    \end{align}
    Since rotation does not change the operator norm, i.e., $\norm{M} = \norm{A_i M A_i^\dag}$, we get $\norm{\Theta[M]} \le 2 \norm{M}$.
\end{proof}

Using these two propositions, we can now upper-bound the operator norm of $D_{\mathbf{W}} \tilde{\rho}_L$.

\begin{proposition}\label{prop:cluster_derivative_norm}
    Given a cluster $\mathbf{W}$ with weight $|\mathbf{W}|=m$, we have
    \begin{equation}
        \norm{D_{\mathbf{W}} \tilde{P}_L} \le (2h^k)^{m}
    \end{equation}
    for all $L \subset [n]$.
\end{proposition}

\begin{proof}
    Suppose there are $m' \le m$ vertices in the first layer of the circuit. For each vertex, the backward lightcone contains at most $k$ qubits initialized in $\ketbra{0}{0}$. Collect all such qubits and denote them by $\ketbra{0}{0}_M$. All other qubits are fully depolarized, so we can use the initial state $h^{|M|} \ketbra{0}{0}_M \otimes I_{M_c}$ instead, where $|M|$ denotes the number of qubits in $M$ and is upper-bounded by $k m'$. Note that the factor of $h^m$ comes in because of the choice of the normalization of $\tilde{P}_L$.
    

    Any non-trivial $D_{\mathbf{W}} \tilde{P}_L$ is obtained by evolving the initial state $h^m \ketbra{0}{0}_M \otimes I_{M_c}$ with a series of unitary operations, $\Theta$, depolarizing channels, and finally dephasing channels. The unitary operations are local and do not change the operator norm. The depolarizing channels and dephasing channels are unital channels, so they do not increase the operator norm due to Proposition~\ref{prop:unital_channel_norm}. The $\Theta$ map increases the operator norm by at most a factor of two because of Proposition~\ref{prop:theta_norm}. Since there are $m$ $\Theta$ maps in total, the total increase is upper-bounded by $2^m$. Therefore, we have
    \begin{equation}
        \norm{D_{\mathbf{W}} \tilde{P}_L} \le 2^m \norm{h^{|M|} \ketbra{0}{0}_M \otimes I_{M_c}} \le (2h^k)^{m}
    \end{equation}
\end{proof}

We now prove Lemma~\ref{lem:cluster_deriv_upper}. The proof is based on the combinatorial estimate visited in~\cite{haah2022optimal,wild2023classical}. Note that the proof is essentially a modification of the proof of Lemma A.6 in~\cite{zhang2025conditional}. The commutation property required there is automatically satisfied in our case because we work with the classical measurement distribution. We will first need to introduce the notions of \emph{graph partition} and \emph{cluster partition}. These definitions are mostly quoted from~\cite{zhang2025conditional} with slight modifications.

Given a cluster $\mathbf{W}$, we can define the interaction graph ${\rm Gra}(\mathbf{W})$ of $\mathbf{W}$. The interaction graph contains $|\mathbf{W}|$ nodes labeled with elements of $\mathbf{W}$. Specifically, if $\mathbf{W} = \{(q_{i,j}, \mu_{i,j})\}$, each node can be labeled by a tuple $((i,j), a)$ where $a$ takes integer value from one to $\mu_{i,j}$. Two nodes $((i,j), a)$ and $((i',j'), a')$ are connected if $(i,j)$ and $(i',j')$ are either the same or are connected in the interaction graph $G$ of the circuit. A \emph{graph partition} $B$ of $\mathbf{W}$ is defined as the graph partition of ${\rm Gra}(\mathbf{W})$. We will mostly consider a special subset of graph partitions where each partition has a connected induced subgraph. We denote ${\rm PaC}(F)$ as the collection of all graph partitions of $F$ such that they partition $F$ into connected induced subgraphs.

Meanwhile, we define a \emph{cluster partition} of $\mathbf{W}$ as a multiset $\mathsf{P}=\{(\mathbf{W}_i,\mu(\mathbf{W}_i))\}$ such that the multiset union gives $\mathbf{W}$. We let $|\mathsf{P}|=\sum_i \mu(\mathbf{W}_i)$ and $\mathsf{P}! = \prod_i \mu(\mathbf{W}_i)!$. One can see that each graph partition $B$ of $\mathbf{W}$ corresponds to a cluster partition $\mathsf{P}$ by simply "forgetting" the $a$ in the node label $((i,j), a)$. On the other hand, for each cluster partition $\mathsf{P}$, there are $\frac{\mathbf{W}!}{\mathsf{P}! \prod_{i}\mathbf{W}_i!}$ graph partitions that correspond to $\mathsf{P}$. Similarly, we denote ${\rm PaC}(\mathbf{W})$ as the collection of all cluster partitions of $\mathbf{W}$ into connected clusters.

To compare, a graph partition is similar to a cluster partition, but when $\mu_{i,j} \ge 1$, then different $\mathbf{W}_i$ in the multiset are treated as distinguishable by assigning labels to each one.

The interaction graph ${\rm Gra}(\mathsf{P})$ of a cluster partition $\mathsf{P}$ is defined as follows: ${\rm Gra}(\mathsf{P})$ contains $|\mathsf{P}|$ nodes corresponding to clusters, and two nodes are connected if and only if their corresponding clusters $\mathbf{W}_i$ and $\mathbf{W}_j$ are connected after taking the union $\mathbf{W}_i \cup \mathbf{W}_j$. The interaction graph ${\rm Gra}(B)$ of a graph partition $B$ is defined similarly. For any graph $F$, let $\chi^*(n,F)$ denote the number of node colorings using exactly $n$ colors such that two connected nodes have different colors.

In the first step of proving Lemma~\ref{lem:cluster_deriv_upper}, we will relate the operator norm of $D_{\mathbf{W}} \log(\tilde{P}_L)$ to a certain graph coloring problem.
\begin{lemma}
    \label{lem:graph_coloring}
    Given any cluster $\mathbf{W}$ with weight $|\mathbf{W}|=m$, we have
    \begin{align}
            \norm{D_{\mathbf{W}} \log(\tilde{P}_L)} \le (2h^k)^{m} \sum_{B  \in {\rm PaC}({\rm Gra}(\mathbf{W}))}   \sum_{n=1}^{|B|} \frac{(-1)^n}{n}\chi^*(n,{\rm Gra}(B))
        \end{align}
\end{lemma}

\begin{proof}
    The proof is similar to the proof of Lemma A.9 in~\cite{zhang2025conditional}. We first consider $D_{\mathbf{W}} \tilde{P}_L$.
    \begin{equation}
        \tilde{P}_L = I + \sum_{k=1}^{\infty} \sum_{\mathbf{W}:|\mathbf{W}|=k} \frac{q_\mathbf{W}}{\mathbf{W}!} D_{\mathbf{W}} \tilde{P}_L
    \end{equation}
    where we organize the cluster expansion by the weight of the clusters. Since $\mathbf{W}$ can be disconnected, we use $\mathsf{P}_{\rm max}(\mathbf{W})$ to denote the maximally connected subset of $\mathbf{W}$, namely the minimal partition that separates $\mathbf{W}$ into connected subsets. For a connected cluster $\mathbf{W}$, $\mathsf{P}_{\rm max}(\mathbf{W})$ is simply $\mathbf{W}$ itself.
    
    Using Proposition~\ref{prop:disconnected_factorize}, we have
    \begin{equation}
        \tilde{P}_L = I + \sum_{k=1}^{\infty} \sum_{\mathbf{W}:|\mathbf{W}|=k} \prod_{\mathbf{V} \in \mathsf{P}_{\rm max}(\mathbf{W})} \left(\frac{q_\mathbf{V}}{\mathbf{V}!} D_{\mathbf{V}} \tilde{P}_L\right)
    \end{equation}
    Notice that because of Proposition~\ref{prop:linear_cluster}, for $D_{\mathbf{V}} \tilde{P}_L$ to be non-zero, $\mu_{i,j}$ must be at most one. This means that $\mathbf{V}!=1$. However, we will still keep the $\mathbf{V}!$ in the denominator to be consistent with the notation in earlier literature.

    Next, we apply the matrix logarithm expansion $\log(I+A)=\sum_{n=1}^{\infty} \frac{(-1)^{n-1}}{n}A^n$ to expand $\log(\tilde{P}_L)$.
    \begin{equation}\label{eq:log_expansion}
        \log(\tilde{P}_L) = \sum_{n=1}^{\infty} \frac{(-1)^{n-1}}{n} \left(\sum_{k=1}^{\infty} \sum_{\mathbf{W}:|\mathbf{W}|=k} \prod_{\mathbf{V} \in \mathsf{P}_{\rm max}(\mathbf{W})} \left(\frac{q_\mathbf{V}}{\mathbf{V}!} D_{\mathbf{V}} \tilde{P}_L\right)\right)^n
    \end{equation}

    To estimate $D_{\mathbf{W}} \log(\tilde{P}_L)$, we need to reorganize the above equation into a cluster expansion of $\log(\tilde{P}_L)$, formally shown below.
    \begin{align}\label{eq:log_expand_reorder}
        \log(\tilde{P}_L) &= \sum_{k=1}^{\infty} \sum_{\mathbf{W} \in \mathcal{G}_k} \sum_{\mathsf{P} \, {\rm partitioning} \, \mathbf{W}}  C(\mathsf{P}) \prod_{\mathbf{V} \in \mathsf{P}} \left(\frac{q_\mathbf{V}}{\mathbf{V}!} D_{\mathbf{V}} \tilde{P}_L\right)
    \end{align}
    where we use $\mathcal{G}_k$ to denote the set of all connected clusters with weight $k$. We only sum over connected clusters because the disconnected clusters have zero cluster derivative. The coefficients $C(P)$ are some coefficients that we will match to Eq.~\eqref{eq:log_expansion}. Note that the ordering of $D_{\mathbf{V}} \tilde{P}_L$ for different $\mathbf{V}$ is not important because we always dephase the density matrix in the end, so they commute. In~\cite{zhang2025conditional}, this commutation relation is not guaranteed and is taken as a technical assumption.

    To reorganize Eq.~\ref{eq:log_expansion} into Eq.~\eqref{eq:log_expand_reorder}, we expand the $n$-th power explicitly.
    \begin{equation}
        \begin{split}
            \log(\tilde{P}_L) 
            =\sum_{n=1}^{\infty} \frac{(-1)^{n-1}}{n}& \sum_{k_1=1}^{\infty}   \sum_{\mathbf{W}_1:|\mathbf{W}_1| =k_1}\sum_{k_2=1}^{\infty}   \sum_{\mathbf{W}_2:|\mathbf{W}_2| =k_2} \ldots \sum_{k_n=1}^{\infty}   \sum_{\mathbf{W}_n:|\mathbf{W}_n| =k_n} \\
            &\prod_{\mathbf{V}_{m,1} \in \mathsf{P}_{\rm max}(\mathbf{W}_1)} \left(\frac{q_{\mathbf{V}_{m,1}}}{\mathbf{V}_{m,1}!}D_{\mathbf{V}_{m,1}}\tilde{P}_L\right)  \\
        \times &\prod_{\mathbf{V}_{m,2} \in \mathsf{P}_{\rm max}(\mathbf{W}_2)}\left(\frac{q_{\mathbf{V}_{m,2}}}{\mathbf{V}_{m,2}!}D_{\mathbf{V}_{m,2}}\tilde{P}_L\right) \\
        \times \ldots  \times  &\prod_{\mathbf{V}_{m,n} \in \mathsf{P}_{\rm max}(\mathbf{W}_n)}\left(\frac{q_{\mathbf{V}_{m,n}}}{\mathbf{V}_{m,n}!}D_{\mathbf{V}_{m,n}}\tilde{P}_L\right)
        \end{split}
        \end{equation}

        Next, we reorganize the above summation by the total cluster weight $|\mathbf{W}|=|\mathbf{W}_1|+|\mathbf{W}_2|+\ldots +|\mathbf{W}_n|$. 
    \begin{equation}
        \begin{split}
            \log(\tilde{P}_L) 
            =\sum_{n=1}^{\infty} \frac{(-1)^{n-1}}{n}& \sum_{k=1}^{\infty}   \sum_{\{\mathbf{W}_i\}:\sum_{i=1}^{n}|\mathbf{W}_i| =k} \\
            &\prod_{\mathbf{V}_{m,1} \in \mathsf{P}_{\rm max}(\mathbf{W}_1)} \left(\frac{q_{\mathbf{V}_{m,1}}}{\mathbf{V}_{m,1}!}D_{\mathbf{V}_{m,1}}\tilde{P}_L\right)  \\
        \times &\prod_{\mathbf{V}_{m,2} \in \mathsf{P}_{\rm max}(\mathbf{W}_2)}\left(\frac{q_{\mathbf{V}_{m,2}}}{\mathbf{V}_{m,2}!}D_{\mathbf{V}_{m,2}}\tilde{P}_L\right) \\
        \times \ldots  \times  &\prod_{\mathbf{V}_{m,n} \in \mathsf{P}_{\rm max}(\mathbf{W}_n)}\left(\frac{q_{\mathbf{V}_{m,n}}}{\mathbf{V}_{m,n}!}D_{\mathbf{V}_{m,n}}\tilde{P}_L\right)
        \end{split}
    \end{equation}
    Now the above equation formally resembles Eq.~\eqref{eq:log_expand_reorder}, where $\mathbf{W} = \cup_i \mathbf{W}_i$ and $\mathsf{P}$ is the graph partition $\{\mathbf{V}_{m,n}\}$. However, multiple terms in the above equation might correspond to the same $\mathbf{W}$ and $P$, so we will have to count the redundancies $C(P)$.

    It turns out that the redundancies are related to the number of graph colorings. First, notice that $\mathbf{V}_{m,n}$ is connected and any $\mathbf{V}_{m,n}$ and $\mathbf{V}_{m',n}$ have to be disconnected if $m \neq m'$. This is because they belong to $\mathsf{P}_{\rm max}(\mathbf{W}_n)$, so $\mathbf{V}_{m,n}$ has to be connected by definition. Meanwhile, if $\mathbf{V}_{m,n}$ and $\mathbf{V}_{m',n}$ are connected, then one can construct a new partition that merges $\mathbf{V}_{m,n}$ and $\mathbf{V}_{m',n}$, thereby violating the condition of being a maximally connected subset. On the other hand, $\mathbf{V}_{m,n}$ and $\mathbf{V}_{m',n'}$ with $n \neq n'$ can in general be connected. The condition that $\mathbf{V}_{m,n}$ being connected implies that $\mathsf{P} := \{\mathbf{V}_{m,n}\} \in {\rm PaC}(\mathbf{W})$. In addition, the condition that $\mathbf{V}_{m,n}$ and $\mathbf{V}_{m',n}$ being disconnected corresponds exactly to the graph coloring condition in ${\rm Gra}(\mathsf{P})$: each value of $n$ is assigned a color and each node in ${\rm Gra}(\mathsf{P})$, labeled by $\mathbf{V}_{m,n}$, is painted in the corresponding color. All nodes in the same color have to be mutually disconnected.

    However, there could be multiple graph colorings of ${\rm Gra}(\mathsf{P})$ that correspond to the same $\mathsf{P}$. This happens when any cluster multiplicity $\mu(\mathbf{W}_i) >1$, as permuting the color within the set of $\mathbf{W}_i$ gives rise to a different graph coloring but corresponds to the same $\mathsf{P}$. Therefore, each cluster partition $\mathsf{P}$ corresponds to $\mathsf{P}!$ graph colorings of ${\rm Gra}(\mathsf{P})$. $C(P)$ then is the sum of all possible numbers of graph colorings of ${\rm Gra}(\mathsf{P})$ with one, two, up to $|P|$ colors, divided by the $\mathsf{P}!$ redundancy, and including the factor of $\frac{(-1)^n}{n}$ originating from the Taylor expansion of the logarithm.
    \begin{equation}
        C(P) = \frac{1}{\mathsf{P}!} \sum_{n=1}^{|P|} \frac{(-1)^n}{n}\chi^*(n,{\rm Gra}(\mathsf{P}))
    \end{equation}
    We can now plug the above equation into Eq.~\eqref{eq:log_expand_reorder} to obtain
    \begin{equation}
        \log(\tilde{P}_L) = \sum_{k=1}^{\infty} \sum_{\mathbf{W} \in \mathcal{G}_k} \sum_{P \in {\rm PaC}(\mathbf{W})}  \left(\frac{1}{\mathsf{P}!} \sum_{n=1}^{|P|} \frac{(-1)^n}{n}\chi^*(n,{\rm Gra}(\mathsf{P}))\right) \prod_{\mathbf{V} \in P} \left(\frac{q_\mathbf{V}}{\mathbf{V}!} D_{\mathbf{V}} \tilde{P}_L\right)
    \end{equation}
    Now we can take the cluster derivative $D_{\mathbf{W}}$ on both sides.
    \begin{equation}
        D_{\mathbf{W}} \log(\tilde{P}_L) = \mathbf{W}! \sum_{P \in {\rm PaC}(\mathbf{W})}  \left(\frac{1}{\mathsf{P}!} \sum_{n=1}^{|P|} \frac{(-1)^n}{n}\chi^*(n,{\rm Gra}(\mathsf{P}))\right) \prod_{\mathbf{V} \in P} \left(\frac{1}{\mathbf{V}!} D_{\mathbf{V}} \tilde{P}_L\right)
    \end{equation}
    Recall that for each cluster partition $P$, there are $\frac{\mathbf{W}!}{\mathsf{P}! \prod_{i}\mathbf{W}_i!}$ graph partitions $B$ that correspond to $P$. Therefore, we can sum over all graph partitions $B$ instead and obtain
    \begin{equation}
        D_{\mathbf{W}} \log(\tilde{P}_L) = \sum_{B \in {\rm PaC}({\rm Gra}(\mathbf{W}))}  \left( \sum_{n=1}^{|P|} \frac{(-1)^n}{n}\chi^*(n,{\rm Gra}(\mathsf{P}))\right) \prod_{\mathbf{V} \in P}  D_{\mathbf{V}} \tilde{P}_L
    \end{equation}
    Finally, we invoke Proposition~\ref{prop:cluster_derivative_norm} to have
    \begin{align}
        \norm{D_{\mathbf{W}} \log(\tilde{P}_L)} &\le \sum_{B \in {\rm PaC}({\rm Gra}(\mathbf{W}))}  \left( \sum_{n=1}^{|P|} \frac{(-1)^n}{n}\chi^*(n,{\rm Gra}(\mathsf{P})\right) \prod_{\mathbf{V} \in P}  \norm{D_{\mathbf{V}} \tilde{P}_L} \\
        &\le \sum_{B \in {\rm PaC}({\rm Gra}(\mathbf{W}))}  \left( \sum_{n=1}^{|P|} \frac{(-1)^n}{n}\chi^*(n,{\rm Gra}(\mathsf{P}))\right)\prod_{\mathbf{V} \in P} h^{|\mathbf{V}|} \\
        &\le (2h^k)^{m} \sum_{B \in {\rm PaC}({\rm Gra}(\mathbf{W}))}  \left( \sum_{n=1}^{|P|} \frac{(-1)^n}{n}\chi^*(n,{\rm Gra}(\mathsf{P}))\right)
    \end{align}
\end{proof}

To derive Lemma~\ref{lem:cluster_deriv_upper} from Lemma~\ref{lem:graph_coloring}, we quote a combinatorial estimate of the graph coloring coefficients in the above equation.

\begin{lemma}[Adapted from Lemma~A.10 in~\cite{zhang2025conditional}]
    \label{lem:combinatorial_estimate}
    For any cluster $\mathbf{W}$, we have 
\begin{align}
    \sum_{B \in {\rm PaC}({\rm Gra}(\mathbf{W}))} \left|\sum_{n=1}^{|B|} \frac{(-1)^n}{n}\chi(n,{\rm Gra}(B)) \right| \le \mathbf{W}! (2e(1+\mathfrak{d}))^{|\mathbf{W}|+1} 
    \end{align}
    
    where ${\rm Gra}(B)$ denotes the induced interaction graph of $B$, $\tau(G)$ denotes the number of spanning trees of $G$, and ${\rm deg}(a)$ denotes the degree of vertex $a$.
\end{lemma}

Finally, we have all the ingredients to prove Lemma~\ref{lem:cluster_deriv_upper}.

\begin{proof}[Proof of Lemma~\ref{lem:cluster_deriv_upper}]
    Let $\mathbf{W}$ be a cluster with weight $|\mathbf{W}|=m$. First, by applying Lemma~\ref{lem:graph_coloring} and then Lemma~\ref{lem:combinatorial_estimate}, we have
    \begin{align}
        \norm{D_{\mathbf{W}} \log(\tilde{P}_L)}
        &\le (2h^k)^{m} \sum_{B  \in {\rm PaC}({\rm Gra}(\mathbf{W}))}   \sum_{n=1}^{|B|} \frac{(-1)^n}{n}\chi^*(n,{\rm Gra}(B))\\
        &\le (2h^k)^{m} \mathbf{W}! (2e(1+\mathfrak{d}))^{m+1},
    \end{align}
    which gives the desired result.
\end{proof}

\section{Approximate Markovianity in Noisy Random Circuits}

\subsection{Proof of Lemma~\ref{lem:trace_distance_bound}}
We first prove the fourth moment bound (Lemma~\ref{lem:trace_distance_bound}) in the main text. We state it here for convenience.
\begin{lemma}
    \label{lem:trace_distance_bound_restate}(restating Lemma~\ref{lem:trace_distance_bound}) Consider the state $\rho_{ABC}$ from a noisy random quantum circuit discussed previously. Then the average trace distance $\bar{D}$ in Eq.~\eqref{eq:bar_D} is upper-bounded by
    \begin{equation}
        \bar{D} \leq h^{3|AC|} \left( \mathbb{E}_{U \sim \mathrm{Haar}} 4 p_0^4 \norm{\rho_{AC|0} - \rho_{A|0} \otimes \rho_{C|0}}_2^2 \right)^{\frac{1}{4}}
    \end{equation}
    Where $|AC|$ is the number of qudits in $AC$, $p_0$ is the probability of measuring the all-zero state on $B$, and $\rho_{AC|0}$ is the post-measurement state conditioned on measuring the all-zero state on $B$. Moreover, the quantity in the parenthesis can be computed using four copies of the state $\tilde{\rho}_{AC|0}$.
    \begin{equation}
        \bar{D} \leq h^{3|AC|} \left( 2 \mathbb{E}_{U \sim \mathrm{Haar}}  \norm{\tilde{\rho}_{AC|0} \otimes \Tr[\tilde{\rho}_{AC|0}] - \Tr_A[\tilde{\rho}_{AC|0}] \otimes \Tr_C[\tilde{\rho}_{AC|0}]}_2^2 \right)^{\frac{1}{4}}
    \end{equation}
\end{lemma}
\begin{proof}
    We start with the definition of the average trace distance $\bar{D}$ in Eq.~\eqref{eq:bar_D}:
    \begin{equation}
        \bar{D} = \mathbb{E}_{U \sim \mathrm{Haar}} \sum_{b} p_b \norm{\rho_{AC|b} -  \rho_{A|b} \otimes \rho_{C|b}}_1
    \end{equation}
    Using the invariance of the Haar measure, we remove the summation over $b$ by fixing $b=0$ and multiplying the result by $h^{|B|}$, where $|B|$ is the number of qudits in $B$. This gives us
    \begin{equation}
        \bar{D} = h^{|B|} \mathbb{E}_{U \sim \mathrm{Haar}}  \norm{\rho_{AC|0} - \rho_{A|0} \otimes \rho_{C|0}}_1
    \end{equation}
    Next, we use the relation between the one-norm and the four-norm
    \begin{equation}
        \bar{D} \le h^{|B|} h^{3 |AC|} \mathbb{E}_{U \sim \mathrm{Haar}} \norm{\rho_{AC|0} - \rho_{A|0} \otimes \rho_{C|0}}_4
    \end{equation}
    Notice that $\norm{A}_4 = \Tr[A A^\dag A A^\dag]^{1/4}$. Since we want a fourth moment quantity averaged over the Haar measure, we apply the Cauchy-Schwarz inequality to move the $1/4$-th power outside the Haar average:
    \begin{equation}
        \bar{D} \le h^{|B|} h^{3 |AC|} \left( \mathbb{E}_{U \sim \mathrm{Haar}} \norm{\rho_{AC|0} - \rho_{A|0} \otimes \rho_{C|0}}_4^4 \right)^{1/4}
    \end{equation}
    Next, we bound the four-norm with the two-norm. 
    \begin{equation}
        \bar{D} \le h^{|B|} h^{3 |AC|} \left( \mathbb{E}_{U \sim \mathrm{Haar}}  p_0^4 \norm{\rho_{AC|0} - \rho_{A|0} \otimes \rho_{C|0}}_2^4 \right)^{1/4}
    \end{equation}
    Now we have the two-norm raised to the fourth power, but we want the two-norm raised to the second power. To remove two powers, we use the purity bound of $\norm{\rho_{AC|0} - \rho_{A|0} \otimes \rho_{C|0}}_2^2$:
    \begin{equation}
        \norm{\rho_{AC|0} - \rho_{A|0} \otimes \rho_{C|0}}_2^2 \le \norm{{\rho_{AC|0}}}_2^2 + \norm{\rho_{A|0} \otimes \rho_{C|0}}_2^2 \le 2
    \end{equation}
    Therefore, we have
    \begin{equation}
        \bar{D} \le h^{|B|} h^{3 |AC|} \left( 2 \mathbb{E}_{U \sim \mathrm{Haar}}  p_0^4 \norm{\rho_{AC|0} - \rho_{A|0} \otimes \rho_{C|0}}_2^2 \right)^{1/4}
    \end{equation}
    Finally, the above bound can be written in terms of the unnormalized state $\tilde{\rho}_{AC|0}$ by noticing that
    \begin{align}
        p_0^4 \norm{\rho_{AC|0} - \rho_{A|0} \otimes \rho_{C|0}}_2^2 &= \Tr\left[\left( p_0^2 \rho_{AC|0} - p_0^2 \rho_{A|0} \otimes \rho_{C|0}\right)^2 \right] \\
        &= \Tr\left[\left( \tilde{\rho}_{AC|0} \otimes \Tr[\tilde{\rho}_{AC|0}] - \Tr_A[\tilde{\rho}_{AC|0}] \otimes \Tr_C[\tilde{\rho}_{AC|0}]\right)^2 \right] \\
        &= \norm{\tilde{\rho}_{AC|0} \otimes \Tr[\tilde{\rho}_{AC|0}] - \Tr_A[\tilde{\rho}_{AC|0}] \otimes \Tr_C[\tilde{\rho}_{AC|0}]}_2^2
    \end{align}
\end{proof}

\subsection{Interpreting $\Delta Z$ as connected correlation}\label{app:conn_corr}
In this section, we show that $\Delta Z$ can be understood as the connected correlation of the statistical mechanics model. Each of the four partition functions $Z_i$ given in Proposition~\ref{prop:fourth_moment_partition} can be understood as the \emph{unnormalized} probability of the given configuration of spins on $A$ and $C$. Specifically, let $P_i$ be the \emph{normalized} probability of the configuration of spins on $A$ and $C$ given by $Z_i$, and let $Z_0 = Z_i / P_i$ be the normalization constant. Then, we can write $\Delta Z$ as
\begin{equation}
    \Delta Z = Z_0 \left( P_1 - P_2 - P_3 + P_4 \right)
\end{equation}
Next, we rewrite the above equation in terms of the connected correlation function of the Potts model. The connected correlation function is defined as
\begin{equation}
    C[i,j] = P(\sigma(i), \sigma(j))_{AC} - P(\sigma(i))_A P(\sigma(j))_C
\end{equation}
where $P(\sigma(i), \sigma(j))_{AC}$ is the joint probability of setting all spins to $\sigma(i)$ on $A$ and $\sigma(j)$ on $C$, and $P(\sigma(i))_A$ and $P(\sigma(j))_C$ are the marginal probabilities on $A$ and $C$ respectively. As an example, $P_1 = P((13)(2)(4), (13)(2)(4))_{AC}$. $\Delta Z$ can be rewritten as a linear combination of connected correlations shown below.

\begin{proposition}
    \label{prop:connected_correlation} $\Delta Z$ can be written as a linear combination of the connected correlation functions, up to a normalization constant $Z_0$.
\begin{equation}
\begin{split}
    &\Delta Z = Z_0 \times \\
        \Big( &C[(13)(2)(4), (13)(2)(4)] - C[(13)(2)(4), (14)(2)(3)] \\
         - &C[(14)(2)(3), (13)(2)(4)] + C[(13)(2)(4), (24)(1)(3)] \Big)
\end{split}
\end{equation}
\end{proposition}
 To arrive at the above proposition, we use the fact that disconnected correlation functions cancel out. This is because the spins on $A$ and $C$ are in the same conjugacy class for all $Z_i$. Because of the permutation symmetry of the four copies, the marginal probabilities are the same.

\begin{proposition}
    \label{prop:marginal_probabilities} If $\sigma(i)$ and $\sigma(j)$ are in the same conjugacy class, then the marginal probabilities on region $L$ are the same: $P(\sigma(i))_L = P(\sigma(j))_L$.
\end{proposition}

As a sanity check, our bound should become vacuous when the noise rate $p$ is small and the circuit depth is above the critical depth for MIE. Indeed, without noise, the Potts model is ferromagnetic and treats all configurations equally. Therefore, as the depth increases, the model enters the ferromagnetic phase and generates a non-trivial connected correlation function at large distance. This is consistent with the onset of MIE.

When the noise is turned on, the Potts model favors the configuration with more spins in the configuration $e=(1)(2)(3)(4)$. This moves the transition to a first-order phase transition, where the connected correlation function decays exponentially even at the critical depth. Therefore, if one can control the constant factor $Z_0$, then one can show that $\bar{D}$ decays exponentially in the distance between $A$ and $C$. 

Unfortunately, this is a non-trivial task when $h$ is constant. The naive $1/h$ expansion gives $Z_0 = O(\exp(n/h))$. In this paper, we focus on the case where $h=\Omega(n)$ so $Z_0 = O(1)$. Improving to $h=O(1)$ would require more fine-grained control of the Potts model. In particular, the $1/h$ expansion corresponds to the domain wall insertion which is structured. This has been exploited in the past in analyzing the two-design time of random quantum circuits. There, the problem maps to the biased random walk of the domain wall, which admits known combinatorial estimates. However, here we face a 24-state Potts model, which contains multiple domain walls, each with different weights and preferred random walk directions. Even worse, domain walls can also merge and split, which makes the combinatorial estimates much more complicated. Therefore, we leave the problem of improving the bound at constant $h$ as an open question.

\subsection{Improving the Dimensionality Constant}
\label{app:improving_dimensionality_constant}
In the main text, the fourth moment bound in Lemma~\ref{lem:trace_distance_bound} has a dimensionality constant $h^{3|AC|}$. This is a rather loose bound. Here we explain how to improve this bound to $h^{d (|\partial_A B| + |\partial_C B|)}$, where $|\partial_A B|$ and $|\partial_C B|$ are the boundary sizes of $B$ connecting $A$ and $C$. This improves the bound when $|AC|$ is large, but the boundary sizes are small, which is the case in our algorithm.

The improvement comes from the observation that only qubits in $AC$ that are near the boundary of $B$ are relevant. We will show how to reduce sites in $A$ to the boundary of $B$. Reducing sites in $C$ goes similarly. We start with $\rho_{AC|b}$ and look at the boundary of $B$ that touches $A$.
\begin{equation}
\includegraphics[width=0.6\linewidth]{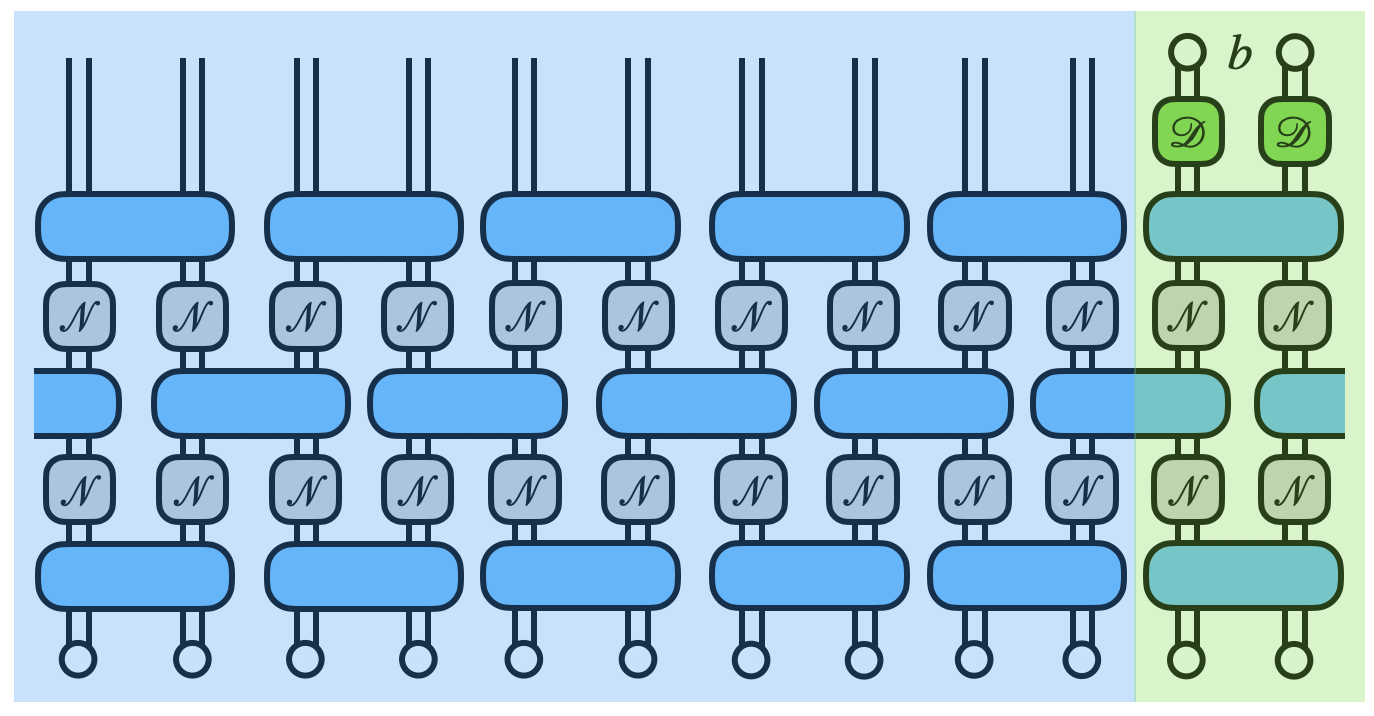}
\end{equation}
Note that we do not dephase $A$. Next, we undo the unitary.
\begin{equation}
\includegraphics[width=0.6\linewidth]{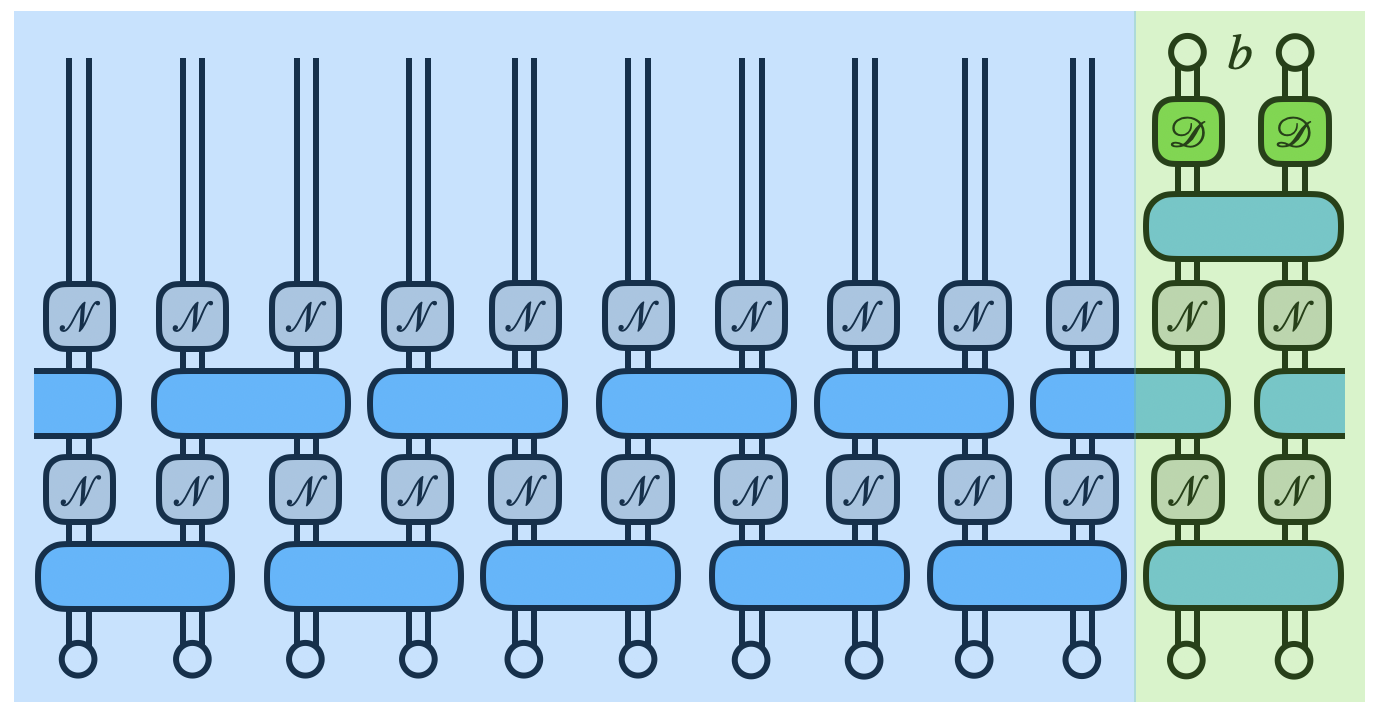}
\end{equation}
Further, we remove the depolarizing channel which can only increase the trace distance.
\begin{equation}
\includegraphics[width=0.6\linewidth]{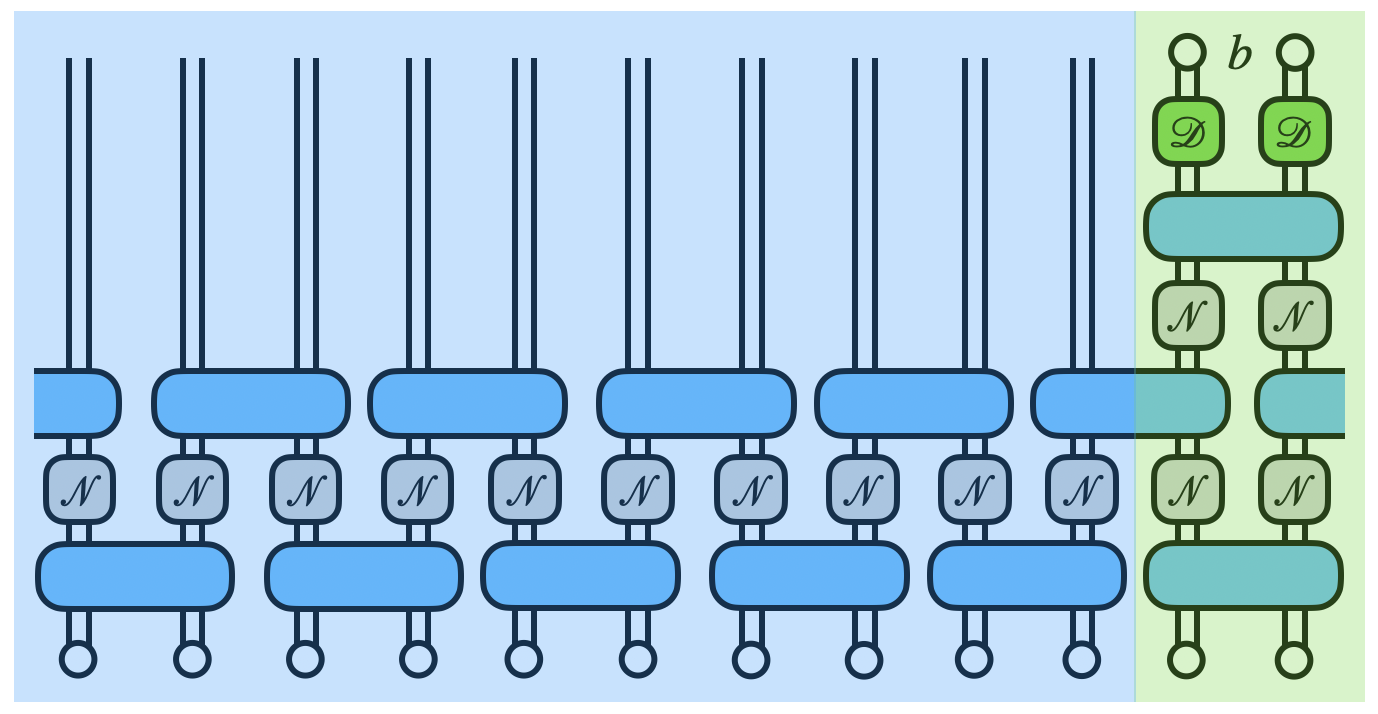}
\end{equation}
Applying this process recursively, we get in the end
\begin{equation}\label{eq:boundary_final}
\includegraphics[width=0.6\linewidth]{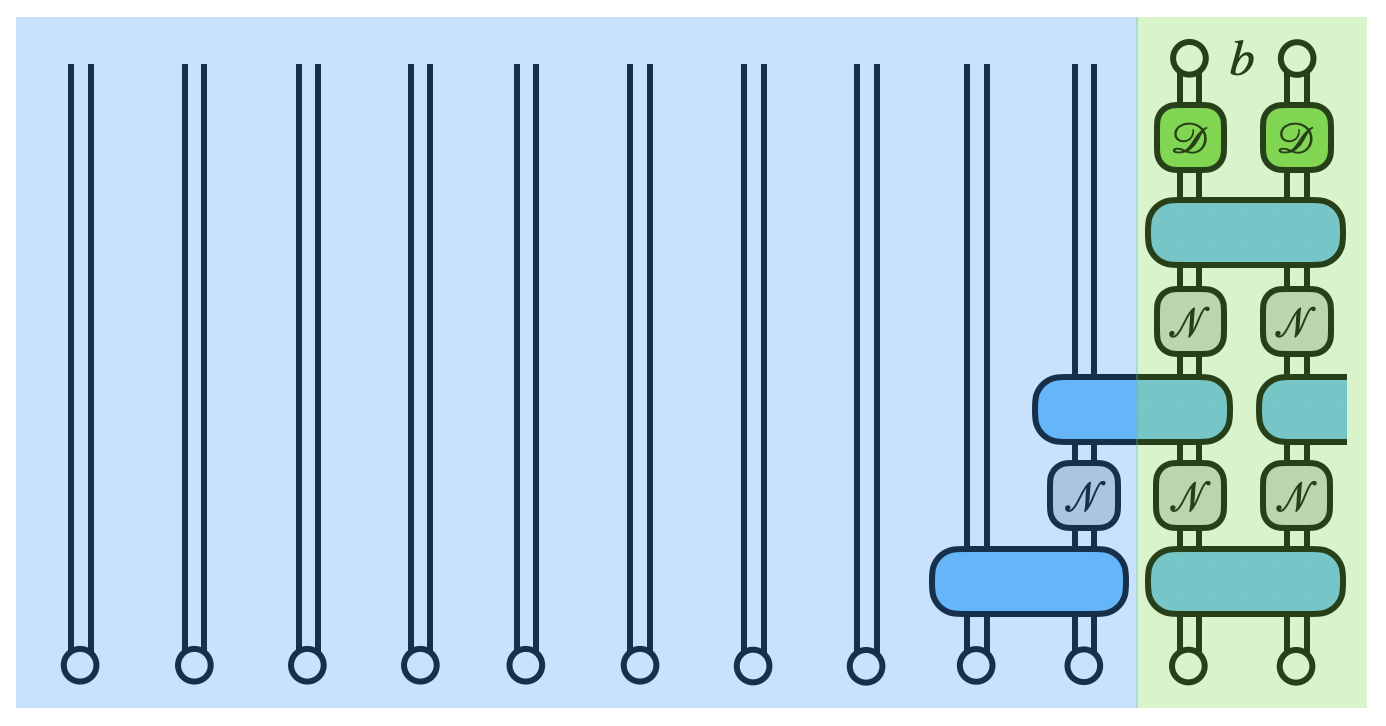}
\end{equation}
We denote the new state $\sigma_{AC|b}$ and will work with it instead. If $\sigma_{AC|b}$ is close to being factorizable, then so is $\rho_{AC|b}$.

\begin{lemma}
If $\norm{\sigma_{AC|b} - \sigma_{A|b} \otimes \sigma_{C|b}}_1 = \epsilon$, then $\norm{\rho_{AC|b} - \rho_{A|b} \otimes \rho_{C|b}}_1 \le \epsilon$
\end{lemma}
\begin{proof}
    The proof is based on doing the above process in reverse. We start with $\sigma_{AC|b}$ and $I_{\max, AC}$ and undo the above process. Applying unitary gates does not change the trace distance, and since gates on $A$ and $C$ are separate, this does not change the factorizability of the right-hand side. Applying local channels does not increase the trace distance and preserves factorizability as well.
\end{proof}
Therefore, we can work with $\sigma_{AC|b}$ instead of $\rho_{AC|b}$. The boundary size of $B$ is at most $d(|\partial_A B| + |\partial_C B|)$, so the dimensionality constant becomes $h^{d(|\partial_A B| + |\partial_C B|)}$.

\subsection{$1/h$ Correction of the Statistical Mechanics Model}\label{app:h_correction}
In this section, we prove the $O(1/h)$ correction term for the statistical mechanics model in Proposition~\ref{prop:partition_function_limit}. For simplicity, we focus on the geometry where $B$ is a $(D+1)$-dimensional slab with dimensions $l_1 = l_{AC}, l_2, \ldots, l_D, d$, and $A, C$ sit on the boundaries along the first dimension. We state the result formally below.

\begin{theorem}\label{thm:bar_D_bound_large_h}
    Under the condition of Theorem~\ref{thm:stat_mech}, suppose $h > c n$ where $c$ is a constant, then $\Delta Z$ is bounded by
    \begin{equation}\label{eq:bar_D_bound_large_h}
        \Delta Z = O\left(\frac{a_p}{|a_h - a_p|} \max{(a_h, a_p)^{l_{AC}}} + a_h^{l_{AC}} \right)
    \end{equation}
    Where $a_h = \frac{n}{ch}$ and $a_p = (1-p')^{\mathcal{A}}$.
\end{theorem}

Let $\Delta T_{ij}^{k}$ be the difference between $T_{ij}^{k}$ and its infinite-$h$ limit.
\begin{equation}
    \Delta T_{ij}^{k} = T_{ij}^{k} - \left((1-p') \delta_{i,j,k} + p' \delta_{i,e} \delta_{j,e} \delta_{k,e}\right)
\end{equation}
Graphically we denote $\Delta T_{ij}^{k}$ as the following light-pink tensor.
\begin{equation}
    \includegraphics[width=0.7\linewidth]{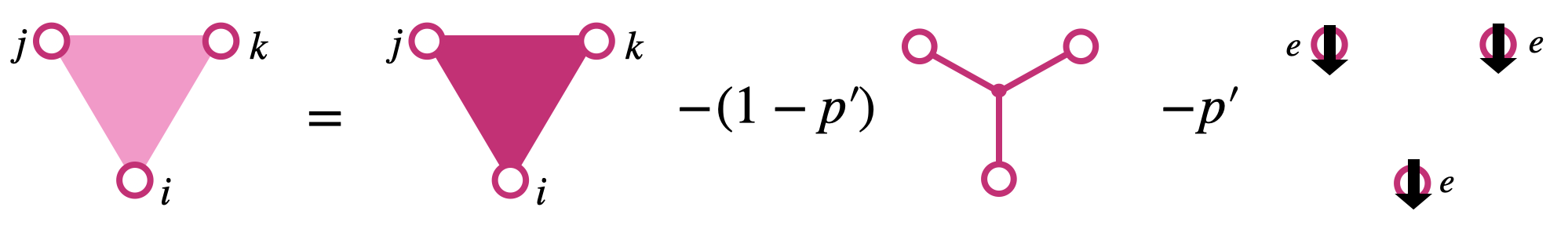}
\end{equation}
We will also use $\bar{T}_{ij}^{k}$ to denote $T_{ij}^{k}$ in the infinite $h$ limit. The term $\Delta T_{ij}^{k}$ can be intuitively understood as a ``domain wall'': in the $h \rightarrow \infty$ limit, $T_{ij}^{k}$ either enforces $i,j,k$ to be identical or pins them to $e$. When $h$ is finite, $\Delta T_{ij}^{k}$ allows $i,j,k$ to differ from each other and from $e$. It is known that $|\Delta T_{ij}^{k}| = O(1/h)$ element-wise~\cite{fisher2022random}. Thus, the formation of a domain wall is suppressed by a factor of $1/h$ for each tensor.

To construct the correction at finite $h$, we sum over all possible tensor networks where some tensors are replaced with $\Delta T_{ij}^{k}$. Let $\kappa$ denote the set of tensors to be replaced, and let $Y_{\kappa}$ be the corresponding partition function with these replacements. Let $|\kappa|$ be the number of tensors replaced. For any $\kappa$, we can bound $Y_{\kappa}$ as follows:
\begin{equation}
    |Y_{\kappa}| \le 24^{3|\kappa|} \max{|\Delta T_{ij}^{k}|}^{|\kappa|} = O\left(\frac{1}{(c h)^{|\kappa|}}\right)
\end{equation}
Where $24^{3|\kappa|}$ comes from the maximal number of indices that are not kept the same by the delta tensor. In the second line, we absorb the constants into $c$. 

Naively, one would sum over all possible $\kappa$ to get the total correction. As long as $h = \Omega(n)$, the sum converges and gives a correction of $O(n/h)$. However, we would also like to show that the correction term decays with distance, with a length scale controlled by the noise rate. To this end, we analyze the structure of $\kappa$ more carefully.

Our key observation is that when $|\kappa|$ is less than $l_{AC}$, $Y_{\kappa}$ becomes highly sensitive to the pinning and typically gives a zero contribution to $\bar{D}$. We first define a property called \emph{infection}. Given a configuration $\kappa$, we select one tensor that is set to $\bar{T}$ and pin the three sites surrounding it to $e$. We say that this tensor is \emph{infectious} if the resulting tensor network factorizes into disconnected components that do not connect $A$ and $C$. We give a simple example below, where by ``infecting'' a site, we completely disconnect the tensor network. Moreover, this tensor network evaluates to zero because of the delta tensor.

\begin{equation}\label{eq:infection_1}
    \includegraphics[width=0.9\linewidth]{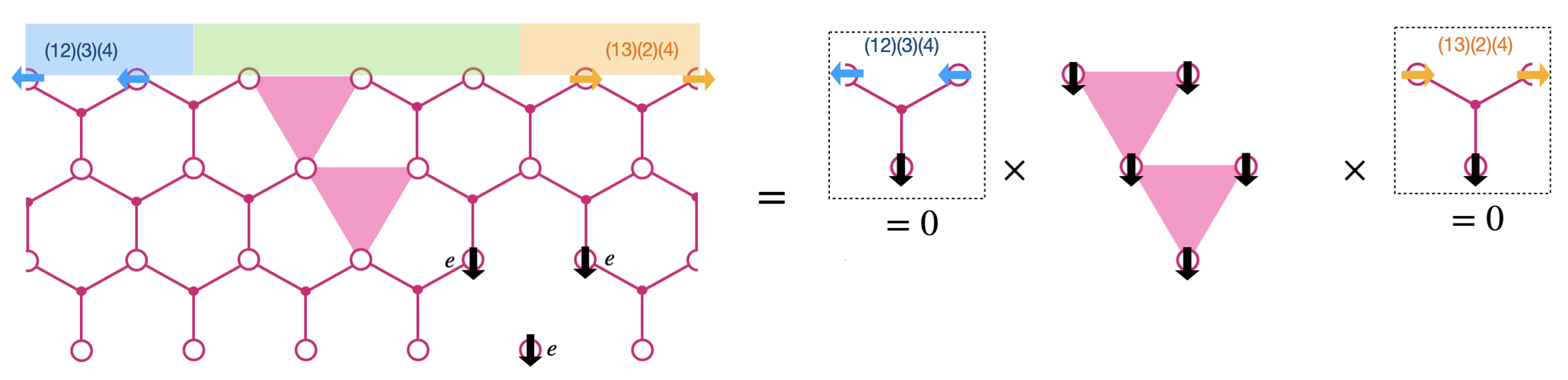}
\end{equation} 

It is easy to see that all $\bar{T}$ tensors are infectious in the above example. In the next example, $\kappa$ contains a string of $\Delta T$ tensors that cuts the tensor network in half. This can be intuitively understood as forming a large domain wall in the middle of the magnet. After infecting one site near $C$, the portion of the tensor network near $A$—which is surrounded by $\Delta T$—remains intact, while the other part becomes ``infected'' and evaluates to zero.

\begin{equation}\label{eq:infection_2}
    \includegraphics[width=0.9\linewidth]{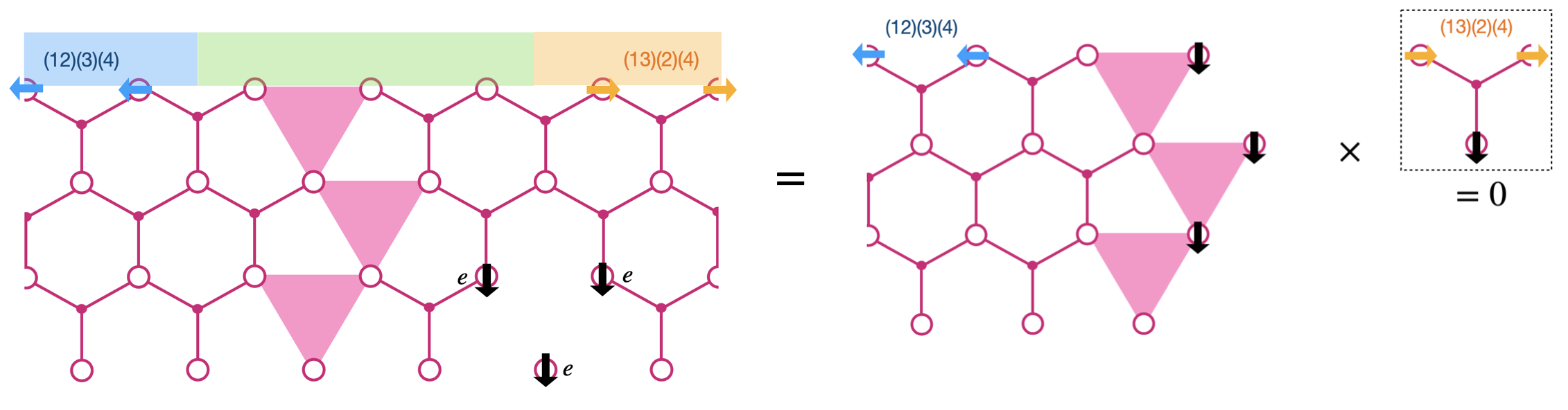}
\end{equation}

One can again see that all $\bar{T}$ tensors are infectious in the above example. Finally, we present a more non-trivial example, where $A$ and $C$ are each individually surrounded by two strings of $\Delta T$ tensors. When a $\bar{T}$ tensor in the middle is infected, it separates the tensor network into two disconnected components, isolating $A$ from $C$. These two components, which may be non-trivial, are denoted as $Z_A$ and $Z_C$.

\begin{equation}\label{eq:infection_3}
    \includegraphics[width=0.9\linewidth]{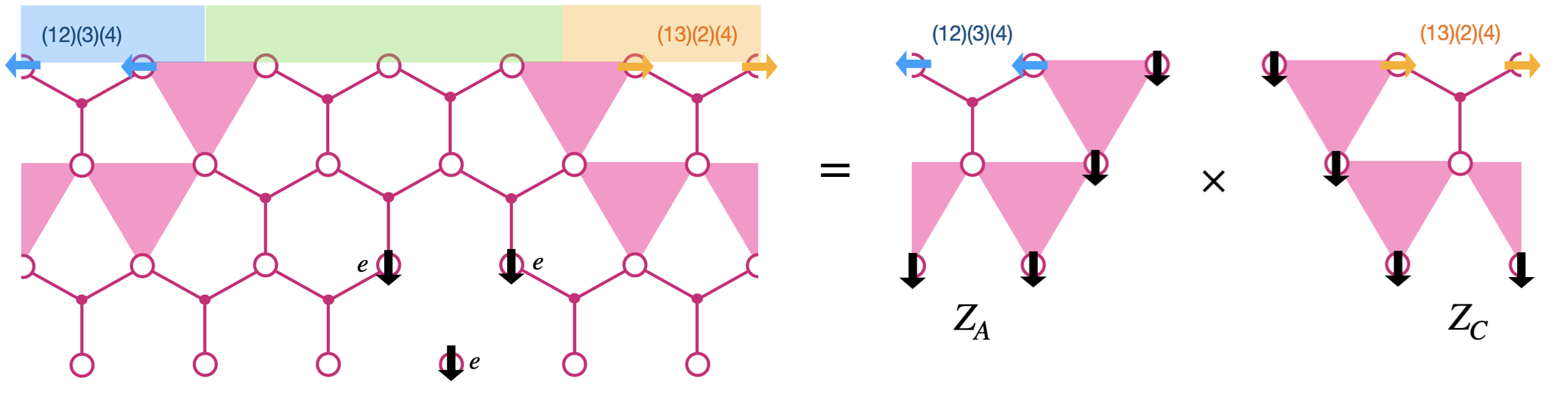}
\end{equation}

In this example, all $\bar{T}$ tensors in the middle are infectious, and the resulting tensor network after infection evaluates to $Z_A \times Z_C$. Finally, we give an example where a string of $\Delta T$ tensors connects $A$ and $C$. In this case, no $\bar{T}$ tensor is infectious.

\begin{equation}\label{eq:non_infectious}
    \includegraphics[width=0.9\linewidth]{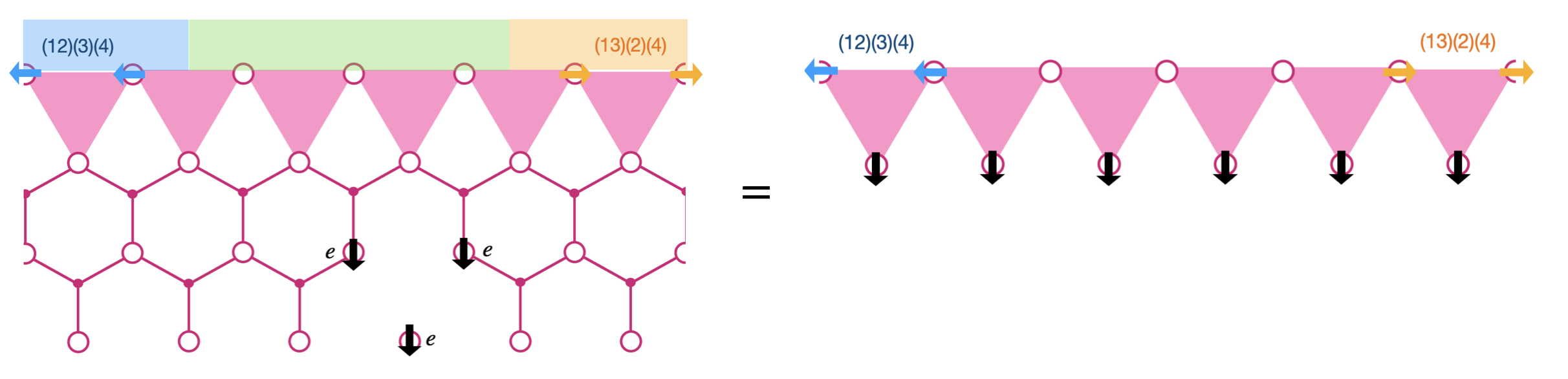}
\end{equation}

Crucially, we show that for a infectious $\bar{T}$, the tensor network does not contribute to $\bar{D}$ after the infection.
\begin{proposition}
    Given a configuration $\kappa$, suppose a tensor $\bar{T}$ is infectious. Then, after infecting the site, the tensor network evaluates either to zero or to $Z_A \times Z_C$, where $Z_A$ and $Z_C$ are supported only on $A$ and $C$, respectively. Moreover, the values of $Z_A$ and $Z_C$ do not depend on the four boundary conditions specified in Proposition~\ref{prop:fourth_moment_partition}. Therefore, their contributions cancel out after taking the linear combination in Eq.~(\ref{eq:fourth_moment_partition}).
    \end{proposition}

The factorization follows directly from the definition of infection. $Z_A$ and $Z_C$ do not depend on the boundary conditions because the pinning is always to elements in the same conjugacy class, and both the tensors and the $e$ elements are invariant under $S_4$ permutations. Using the intuition from~\ref{app:conn_corr}, the above proposition says that pinning an infectious tensor kills connected correlation.

At this point, we have shown that if an infectious tensor is pinned to $e$, then the corresponding tensor network does not contribute to $\bar{D}$. We have also seen that when $|\kappa|$ is at least $l_{AC}$, it is possible for no $\bar{T}$ tensor to be infectious, as illustrated in Eq.~(\ref{eq:non_infectious}). To demonstrate the decay of the $O(1/h)$ correction with distance, we will show that if $|\kappa| < l_{AC}$, then most $\bar{T}$ tensors are infectious.

\begin{lemma}\label{lem:noninfectious_bound}
    Consider any $\kappa$ with $|\kappa| = m < l_{AC}$. The number of non-infectious $\bar{T}$ tensors, denoted as $I_{\kappa}$, is upper-bounded by $\mathcal A m$.
\end{lemma}

\begin{proof}
    Consider the $D$ dimensional hyperplane that is perpendicular to the dimension that separates $A$ and $C$ (in other words they are cross sections of the slab). There are $l_{AC}$ such hyperplanes. If there are $m$ $\Delta T$ tensors, then there are at least $l_{AC} - m$ hyperplanes that do not contain any $\Delta T$ tensor. Any $\bar{T}$ tensor on these hyperplanes is infectious because it can be connected to both boundaries without crossing any $\Delta T$ tensor. Therefore, the number of non-infectious $\bar{T}$ tensors is upper-bounded by $\mathcal A m$.
\end{proof}

The above lemma shows that when $m < l_{AC}$, most of the $\bar{T}$ tensors are infectious. Since each tensor is independently pinned with probability $1-p'$, the probability that no infectious tensor is pinned is $(1-p')^{\mathcal{A} l_{AC}- I_{\kappa}}$. Therefore, the contribution from $Y_{\kappa}$ to $\bar{D}$ is upper-bounded by
\begin{equation}
    |Y_{\kappa}| = O\left((1-p')^{\mathcal{A}l_{AC} - I_{\kappa}} \frac{1}{(c h)^{|\kappa|}}\right)
\end{equation}

We are now ready to compute the $O(1/h)$ correction of $\Delta Z$. 

\begin{proof}
    (Proof of Theorem~\ref{thm:bar_D_bound_large_h}). To bound the contribution at a finite $h$, we sum over all such $Y_{\kappa}$.
\begin{equation}
    \sum_{\kappa} Y_{\kappa} = \sum_{m=0}^{n} \sum_{\kappa: |\kappa|=m} Y_{\kappa}
\end{equation}
Apply the bound on $|Y_{\kappa}|$, we have
\begin{equation}
    \sum_{\kappa} Y_{\kappa} = \sum_{m=0}^{n} \binom{n}{m} O\left((1-p')^{\mathcal{A}l_{AC}- I_{\kappa}} \frac{1}{(c h)^{|\kappa|}}\right)
\end{equation}
Next we bound $I_{\kappa}$ using Lemma~\ref{lem:noninfectious_bound}. We focus on the case where $l_{AC} > d$ and divide the summation into two parts.
\begin{equation}
    \begin{split}
        \sum_{\kappa} Y_{\kappa} & = \sum_{m=0}^{l_{AC}-1} \binom{n}{m} O\left((1-p')^{(l_{AC} - m)\mathcal{A}} \frac{1}{(c h)^{m}}\right) \\
        & + \sum_{m=l_{AC}}^{n} \binom{n}{m} O\left(\frac{1}{(c h)^{m}}\right)
    \end{split}
\end{equation}
Where in the first line we apply the bound $I_{\kappa} \le m \mathcal{A}$ in Lemma~\ref{lem:noninfectious_bound} and then use Proposition~\ref{prop:fourth_moment_partition} to suppress the contribution from infectious tensors. In the second line we face configurations like Eq.~(\ref{eq:non_infectious}) where no $\bar{T}$ tensor is infectious, so we simply ignore the pinning effect.

The second line is easy to evaluate: each term is bounded by $a_h = \frac{n}{ch}$. Therefore, as long as $h=\Omega(n)$ with a sufficiently large constant, the second term decays as $a_h^{l_{AC}}$.
\begin{equation}
    \text{second line} = O(a_h^{l_{AC}})
\end{equation}

Let $a_p = (1-p')^{\mathcal A}$. We can rewrite the first line as
\begin{equation}
    \text{first line} \le C \sum_{m=0}^{l_{AC}-1} a_h^{m} a_p^{l_{AC}-m} = C a_p \frac{a_p^{l_{AC}-1} - a_h^{l_{AC}-1}}{a_p - a_h}
\end{equation}
We loosely bound the summand by
\begin{equation}
    \text{first line}  = O\left(\frac{a_p}{|a_h - a_p|} \max{(a_h, a_p)^{l_{AC}}}\right)
\end{equation}
Putting the two lines together, we arrive at our desired result.
\begin{equation}
    \sum_{\kappa} Y_{\kappa} = O\left(\frac{a_p}{|a_h - a_p|} \max{(a_h, a_p)^{l_{AC}}} + a_h^{l_{AC}} \right)
\end{equation}
One can check that when $p=0$, the first term does not decay with $l_{AC}$, while turning on any non-zero $p$ makes the first term decay exponentially with $l_{AC}$.
\end{proof}

\subsection{Details of the Clifford Numerics}\label{app:clifford_detail}

We provide the details of the Clifford numerics in this section. We consider a two-dimensional rectangular geometry with nearest-neighbor two-qubit gates. We alternate between horizontal and vertical two-qubit gates. We choose the tripartition such that qubits in the left $d$ columns belong to $\textcolor{Cerulean}{A}$, qubits in the right $d$ columns belong to $\textcolor{Goldenrod}{C}$, and qubits in the middle $l_{AC}$ columns belong to $\textcolor{Green}{B}$. The reason we choose $d$ columns in $\textcolor{Cerulean}{A}$ and $\textcolor{Goldenrod}{C}$ is that adding more columns cannot increase the trace distance, as observed in Appendix~\ref{app:improving_dimensionality_constant}. We use 10 rows of qubits throughout the simulations.

We perform the Clifford simulation in a Monte-Carlo sampling fashion. For each shot, we randomly choose an instance of the two-qubit Clifford gates. Also, for each spacetime location, we trace out the qubit with a probability of $p$. We perform the Clifford simulation, perform measurement on $\textcolor{Green}{B}$ in the end, compute the trace distance of the post-measurement state to the product state, and average over shots to compute $\bar{D}$. The trace distance between two stabilizer states is analytic and can be computed following the methods in the literature.

Note that given a circuit instance, different measurement outcomes correspond to different mixed stabilizer states with the same Pauli operators in the stabilizer generators, but they have different signs. Since the measurement outcome probability and the trace distance to the product states are both independent of the signs of the Pauli operators, we simply post-select to measuring $\ket{0}$ on $\textcolor{Green}{B}$ in practice. We use QuantumClifford.jl for all Clifford numerics.

\end{widetext}

\bibliography{apssamp}

\end{document}